\documentclass[11pt]{article}
\usepackage{latexsym}
\usepackage{tabularx,booktabs,multirow,delarray,array}
\usepackage{graphicx,amssymb,amsmath,amssymb}

\newcommand{\bcls}{\mbox{BCLS}}

\def\calL{\mathcal{L}}
\def\calI{\mathcal{I}}

\begin{document}

\baselineskip=14.0pt

\title{
\vspace*{-0.55in} Algorithms on Minimizing the Maximum Sensor Movement
for Barrier Coverage of a Linear
Domain}

\author{
Danny Z. Chen\thanks{Department of Computer Science and Engineering,
University of Notre Dame, Notre Dame, IN 46556, USA.
E-mail: {\tt dchen@nd.edu}. The research of Chen was supported in part by NSF under Grants CCF-0916606 and CCF-1217906.
}
\hspace*{0.2in} Yan Gu\thanks{Department of Computer Science and Technology, Tsinghua University, Beijing 100084, China. E-mail: {\tt henryy321@gmail.com}.}
\hspace*{0.2in} Jian Li\thanks{Institute for Interdisciplinary Information Sciences (IIIS), Tsinghua University, Beijing 100084, China. E-mail: {\tt lijian83@mail.tsinghua.edu.cn}.}
\hspace*{0.2in} Haitao Wang\thanks{Corresponding author. Department of Computer Science, Utah
State University, Logan, UT 84322, USA. E-mail: {\tt
haitao.wang@usu.edu}.
This work was partially done while the
author was visiting IIIS at Tsinghua University.}
}

\date{}

\maketitle

\thispagestyle{empty}

\newtheorem{lemma}{Lemma}
\newtheorem{theorem}{Theorem}
\newtheorem{corollary}{Corollary}
\newtheorem{fact}{Fact}
\newtheorem{definition}{Definition}
\newtheorem{observation}{Observation}
\newtheorem{condition}{Condition}
\newtheorem{property}{Property}
\newtheorem{claim}{Claim}
\newenvironment{proof}{\noindent {\textbf{Proof:}}\rm}{\hfill $\Box$
\rm}
\newtheorem{Prob}{Problem}
\newtheorem{Theo}{Theorem}
\newtheorem{Lem}{Lemma}
\newtheorem{Coro}{Corollary}
\newtheorem{Def}{Definition}
\newtheorem{Obs}{Observation}

\pagestyle{plain}
\pagenumbering{arabic}
\setcounter{page}{1}

\vspace*{-0.3in}
\begin{abstract}
In this paper, we study the problem of moving $n$ sensors on a line
to form a barrier coverage of a specified segment of the line such that the
maximum moving distance of the sensors is minimized.
Previously, it was an open question whether this problem on
sensors with arbitrary sensing ranges is solvable in
polynomial time. We settle this open question positively by giving an
$O(n^2\log n)$ time algorithm.
For the special case when all sensors have the same-size sensing range,
the previously best solution takes $O(n^2)$ time. We present an
$O(n\log n)$ time algorithm for this case; further, if all sensors are
initially located on the coverage segment, our algorithm takes
$O(n)$ time. Also, we extend
our techniques to the cycle version of the problem where the barrier coverage
is for a simple cycle and the sensors are allowed to move only along
the cycle. For sensors with the same-size
sensing range, we solve the cycle version in $O(n)$ time, improving
the previously best $O(n^2)$ time solution.
\end{abstract}

\section{Introduction}

A Wireless Sensor Network (WSN) uses a large number of sensors to
monitor some surrounding environmental phenomena
\cite{ref:AkyildizWi02}.
Each sensor is equipped with a sensing device with limited battery-supplied energy.
The sensors process data obtained and forward the data to a base station.
Intrusion detection and border surveillance constitute a
major application category for WSNs. A
main goal of these applications is to detect intruders as they
cross the boundary of a region or domain. For example, research efforts were made to
extend the scalability of WSNs to the monitoring of international
borders \cite{ref:Hu08,ref:KumarBa07}.
Unlike the traditional {\it full coverage}
\cite{ref:LiLo08,ref:YangSc07,ref:ZouA05} which requires an entire target
region to be covered by the sensors,
the {\it barrier coverage}
\cite{ref:BhattacharyaOp09,ref:ChenDe07,ref:CzyzowiczOn09,ref:CzyzowiczOn10,ref:KumarBa07}
only seeks to cover the perimeter of the region to ensure that any
intruders are detected as they cross the region border.
Since barrier coverage requires fewer sensors, it is often preferable to
full coverage. Because sensors have limited
battery-supplied energy, it is desired to minimize their movements.
In this paper, we study a one-dimensional barrier coverage problem
where the barrier is for a (finite) line segment and the sensors are
initially located on the line containing the barrier segment
and allowed to move on the line. As discussed in the previous work
\cite{ref:CzyzowiczOn09,ref:CzyzowiczOn10,ref:MehrandishMi11}
and shown in this paper, barrier coverage even for 1-D domains
poses some challenging algorithmic issues. Also, our 1-D solutions
may be used in solving more general problems. For example, if
the barrier is sought for a simple polygon, then we may consider
each of its edges separately and apply our algorithms to each edge.

In our problem, each sensor has a {\em sensing range} (or
{\em range} for short) and we want to move the sensors to form a
coverage for the barrier such that the maximum sensor movement is
minimized. We present efficient algorithms for this problem,
improving the previous work and settling an open question.
Also, we extend our techniques to the cycle version where the
barrier is for a simple cycle and the sensors are allowed to move only along
the cycle.

\subsection{Problem Definitions, Previous Work, and Our Results}

Denote by $B=[0,L]$ the barrier that is a line segment from $x=0$ to
$x=L>0$ on the $x$-axis. A set $S=\{s_1,s_2,\ldots,s_n\}$
of $n$ mobile sensors is initially located on the $x$-axis. Each sensor
$s_i\in S$ has a range $r_i>0$ and is located at the coordinate $x_i$ of
the $x$-axis.
We assume $x_1\leq x_2\leq \cdots\leq x_n$. If a sensor $s_i$
is at the position $x'$, then we say $s_i$ {\em covers}
the interval $[x'-r_i,x'+r_i]$, called the
{\em covering interval} of $s_i$. Our problem is to find a set of
{\em destinations} on the $x$-axis, $\{y_1,y_2,\ldots,y_n\}$, for the sensors
(i.e., for each $s_i\in S$, move $s_i$ from $x_i$ to $y_i$)
such that each point on the barrier $B$ is covered by at least one sensor
and the maximum moving distance of
the sensors (i.e., $\max_{1\leq i\leq n}\{|x_i-y_i|\}$) is minimized.
We call this problem the {\em barrier coverage on a line segment}, denoted
by $\bcls$.  We assume $2\cdot \sum_{i=1}^nr_i\geq L$
(otherwise, a barrier coverage for $B$ is not possible).

The {\em decision version} of \bcls\ is defined as follows.
Given a value $\lambda\geq 0$, determine whether there is a
{\em feasible solution} in which each point of $B$ is covered by at least one sensor and the moving distance of each sensor is at
most $\lambda$. The decision version characterizes a problem model
in which the sensors have a limited energy and we want to know whether
their energy is sufficient to move and form a barrier coverage.

If the ranges of all sensors are the same (i.e., the $r_i$'s are all equal),
then we call it the {\em uniform case} of \bcls. When the sensors have
arbitrary ranges, we call it the {\em general case}.

The \bcls\ problem has been studied before. The uniform case has been
solved in $O(n^2)$ time \cite{ref:CzyzowiczOn09}. An $O(n)$ time algorithm is also given in \cite{ref:CzyzowiczOn09} for the decision version of the uniform case.
However, it has been open whether the general case is solvable in
polynomial time \cite{ref:CzyzowiczOn09}.

In this paper, we settle the open problem on the general \bcls,
presenting an $O(n^2\log n)$ time algorithm for it.
We also solve the decision version of the general \bcls\
in $O(n\log n)$ time. Since this is a basic
problem on sensors and intervals and our algorithm is the
first-known polynomial time solution for it, we expect our results and
techniques to be useful for other related problems. Further, for the uniform
case, we derive an $O(n\log n)$ time algorithm, improving the
previous $O(n^2)$ time solution \cite{ref:CzyzowiczOn09}; if all
sensors are initially on $B$, our algorithm runs in $O(n)$ time.

In addition, we consider the {\em simple cycle barrier coverage}
where the barrier is represented as a simple cycle and the $n$
sensors are initially on the cycle and are allowed to move only along
the cycle. The goal is to move the sensors
to form a barrier coverage and minimize the maximum sensor movement.
If all sensors have the same range,
Mehrandish \cite{ref:MehrandishOn11} gave an $O(n^2)$ time algorithm,
and we present an $O(n)$ time solution in this paper.

\subsection{Related Work}

Besides the results mentioned above, an $O(n)$ time $2$-approximation
algorithm for the uniform \bcls\ was also given in \cite{ref:CzyzowiczOn09} and a
variation of the decision version of the general \bcls, where one sensor is required to be in a given position in the final solution, is shown to be NP-hard \cite{ref:CzyzowiczOn09}; however, it is not known whether the general \bcls\ is NP-hard. Further, in the special case where the {\em order preserving} property holds, i.e., $y_1\leq y_2\cdots\leq y_n$ holds in an optimal solution, polynomial time algorithms exist for the general \bcls\  \cite{ref:CzyzowiczOn09}.
Additional results were also given in \cite{ref:CzyzowiczOn09} for
the case $2\cdot\sum_{i=1}^nr_i<L$
(although in this case $B$ cannot be entirely covered).

Mehrandish {\em et al.}~\cite{ref:MehrandishOn11,ref:MehrandishMi11}
also considered the line segment barrier, but unlike the \bcls\ problem,
they intended to use the minimum number of
sensors to form a barrier coverage, which they proved to be
NP-hard. But, if all sensors have the same range, polynomial time
algorithms were possible \cite{ref:MehrandishOn11,ref:MehrandishMi11}.
Another study of the line segment barrier \cite{ref:CzyzowiczOn10}
aimed to minimize the sum of the moving distances of all
sensors; this problem is NP-hard \cite{ref:CzyzowiczOn10}, but
is solvable in polynomial time when all sensors have the same
range \cite{ref:CzyzowiczOn10}. In addition, Li {\em et
al.}~\cite{ref:LiMi11} considers the linear coverage problem which
aims to set an energy for each sensor to form a coverage such that the
cost of all sensors is minimized. There \cite{ref:LiMi11}, the sensors
are not allowed to move, and the more energy a sensor has, the larger
the covering range of the sensor and the larger the cost of the
sensor.

Bhattacharya {\em et al.}~\cite{ref:BhattacharyaOp09} studied a
2-D barrier coverage problem in which the barrier is a
circle and the sensors, initially located inside the circle, are
moved to the circle to form a coverage such that
the maximum sensor movement is minimized; the ranges of the
sensors are not explicitly specified but the destinations of the sensors are
required to form a regular $n$-gon on the circle. Subsequent improvements
of the results in \cite{ref:BhattacharyaOp09} have been
made \cite{ref:ChenOP11,ref:TanNe10}.
In addition, Bhattacharya {\em et al.}~\cite{ref:BhattacharyaOp09}
presented some results on the corresponding
min-sum problem version (minimizing the sum of the
moving distances of all sensors); further improvement was
also given in \cite{ref:ChenOP11,ref:TanNe10}.

Some other barrier coverage problems have been studied.
For example, Kumar {\em et al.}~\cite{ref:KumarBa07} proposed
algorithms for determining whether a region is barrier covered after
the sensors are deployed. They considered both the deterministic
version (the sensors are deployed deterministically) and the
randomized version (the sensors are deployed randomly), and aimed
to determine a barrier coverage with high probability.
Chen {\it et al.}~\cite{ref:ChenDe07} introduced a local barrier coverage
problem in which individual sensors determine the barrier coverage locally.

\subsection{An Overview of Our Approaches}

For any problem we consider, let $\lambda^*$ denote the maximum sensor
movement in an optimal solution.

For the uniform \bcls, as shown in \cite{ref:CzyzowiczOn09},
a useful property is that there always exists an
{\em order preserving} optimal solution, i.e., the order of the
sensors in the optimal solution is the same as that in the input.
Based on this property, the previous $O(n^2)$ time algorithm
\cite{ref:CzyzowiczOn09} covers $B$ from left to right;
in each step, it picks the next sensor and re-balances the current maximum
sensor movement. In this paper, we take a very different approach.
With the order preserving
property, we determine a set $\Lambda$ of candidate values for
$\lambda^*$ with $\lambda^*\in \Lambda$. Consequently, by using
the decision algorithm, we can find $\lambda^*$ in $\Lambda$. But,
this approach may be inefficient since $|\Lambda|=\Theta(n^2)$.
To reduce the running time, our
strategy is not to compute the set $\Lambda$ explicitly.
Instead, we compute an element in $\Lambda$ whenever we need it. A
possible attempt would be to first find a sorted order for
the elements of $\Lambda$ or (implicitly) sort the elements of
$\Lambda$, and then obtain $\lambda^*$ by binary search.
However, it seems not easy to (implicitly) sort the elements of $\Lambda$.
Instead, based on several new observations, we manage to find a way to
partition the elements of $\Lambda$ into $n$ sorted lists, each list containing
$O(n)$ elements. Next, by using a technique
called {\em binary search on sorted arrays} \cite{ref:ChenRe11}, we are able to
find $\lambda^*$ in $\Lambda$ in $O(n\log n)$ time.
For the special case when all sensors are initially located on $B$, a
key observation we make is that $\lambda^*$ is precisely the maximum
value of the candidate set $\Lambda$. Although $\Lambda=\Theta(n^2)$,
based on new observations, we show that its maximum value can be computed in $O(n)$ time.

For the general \bcls, as indicated in \cite{ref:CzyzowiczOn09}, the
order preserving property no longer holds. Consequently, our
approach for the uniform case does not work. The main difficulty of this
case is that we do not know the order of the sensors
appeared in an optimal solution. Due to this difficulty, no polynomial
time algorithm was known before for the general \bcls. To solve this problem,
we first develop a greedy algorithm for the decision version of the general \bcls.
After $O(n\log n)$ time preprocessing, our decision algorithm takes
$O(n)$ time for any value $\lambda$. If $\lambda\geq \lambda^*$,
implying that there exists a feasible solution, then our decision
algorithm can determine the order of sensors in a feasible solution for
covering $B$. For the general \bcls, we seek to simulate the
behavior of the decision algorithm on $\lambda=\lambda^*$. Although we do
not know the value $\lambda^*$, our algorithm determines the same sensor
order as it would be obtained by the decision algorithm on the
value $\lambda=\lambda^*$. To this end, each step of the algorithm uses
our decision algorithm as a decision procedure. The idea is somewhat
similar to parametric search \cite{ref:ColeSl87,ref:MegiddoAp83}, and here
we ``parameterize" our decision algorithm. However, we should point out a few
differences. First, unlike the typical parametric search \cite{ref:ColeSl87,ref:MegiddoAp83},
our approach does not involve any parallel scheme and is practical.
Second, normally, if a problem can be solved by parametric search, then
there also exist other (simpler) polynomial time algorithms for the problem
although they might be less efficient than the parametric search solution
(e.g., the slope selection problem \cite{ref:ColeAn89}). In contrast, for
our general \bcls\ problem, so far we have not found any other (even straightforward)
polynomial time algorithm.

In addition, our $O(n)$ time algorithm for the simple cycle barrier coverage
is a generalization of our approach for the special case of the uniform \bcls\
when all sensors are initially located on $B$.

For ease of exposition, we assume that initially no two sensors are located
at the same position (i.e., $x_i\neq x_j$ for any $i\neq j$), and the covering
intervals of any two different sensors do not share a common endpoint.
Our algorithms can be easily generalized to the general situation.

The rest of the paper is organized as follows. In Section
\ref{sec:general}, we describe our algorithms for the general \bcls. In
Section \ref{sec:uniform}, we present our algorithms for the uniform \bcls.
Our results for the simple cycle barrier coverage are discussed in
Section \ref{sec:circle}. Section \ref{sec:conclusions} finally concludes the paper.

\section{The General Case of \bcls}
\label{sec:general}

In this section, we present our algorithms for the general \bcls\
problem. Previously, it was an open problem whether the general \bcls\ can be
solved in polynomial time. The main difficulty is that we do not know the
order of the sensors in an optimal solution. Our main effort is for resolving
this difficulty, and we derive an $O(n^2\log n)$ time algorithm for
the general \bcls.

We first give our algorithm for the decision version (in Section
\ref{sec:decision}), which is crucial for solving the general \bcls\
(in Section \ref{sec:algogeneral}) that we refer to as the {\em
optimization version} of the problem.


For each sensor $s_i\in S$, we call the right (resp., left) endpoint
of the covering interval of $s_i$ the {\em right}
(resp., {\em left}) {\it extension} of $s_i$. Each of the right and
left extensions of $s_i$ is an {\em extension} of $s_i$.
Denote by $p(x')$ the point on the $x$-axis whose coordinate is $x'$,
and denote by $p^+(x')$ (resp., $p^-(x')$) a point to the right
(resp., left) of $p(x')$ and infinitely close to $p(x')$. The concept
of $p^+(x')$ and $p^-(x')$ is only used to explain the algorithms, and
we never need to find such a point explicitly in the algorithm.
Let $\lambda^*$ denote the maximum sensor moving distance in an optimal
solution for the optimization version of the general \bcls\ problem. Note that
we can easily determine whether $\lambda^*=0$,
say, in $O(n\log n)$ time. Henceforth, we assume $\lambda^*> 0$.

\subsection{The Decision Version of the General \bcls}
\label{sec:decision}

Given any value $\lambda$, the decision version is to determine whether
there is a feasible solution in which the maximum sensor movement is at
most $\lambda$. Clearly, there is a feasible solution if and only
if $\lambda\geq \lambda^*$. We show that after $O(n\log n)$
time preprocessing, for any $\lambda$, we can determine whether
$\lambda\geq \lambda^*$ in $O(n)$ time.
We explore some properties of a feasible solution in Section
\ref{subsec:preli}, describe our decision algorithm in Section
\ref{subsec:desc}, argue its correctness in Section
\ref{subsec:correct}, and discuss its implementation in
Section \ref{subsec:imple}. In Section \ref{subsec:more}, we show that
by extending the algorithm, we can also determine whether
$\lambda>\lambda^*$ in the same time bound; this is particularly
useful to our optimization algorithm in Section \ref{sec:algogeneral}.

\subsubsection{Preliminaries}
\label{subsec:preli}

By a sensor {\em configuration}, we refer to a specification of where
each sensor $s_i\in S$ is located. By this definition, the input is a
configuration in which each sensor $s_i$ is located at $x_i$.
The {\em displacement} of a sensor in a configuration $C$ is the distance
between the position of the sensor in $C$ and its original position in the
input.  A configuration $C$ is a {\em feasible solution} for the distance
$\lambda$ if the sensors in $C$ form a barrier coverage of $B$ (i.e.,
the union of the covering intervals of the sensors in $C$ contains $B$)
and the displacement of each sensor is at most $\lambda$.
In a feasible solution, a subset $S'\subseteq S$ is called a {\em
solution set} if the sensors in $S'$ form a barrier coverage;
of course, $S$ itself is also a solution
set. A feasible solution may have multiple solution sets. A sensor $s_i$
in a solution set $S'$ is said to be {\em critical} with respect to
$S'$ if $s_i$ covers a
point on $B$ that is not covered by any other sensor in $S'$.
If every sensor in $S'$ is critical, then $S'$ is called a {\em
critical set}.

Given any value $\lambda$, if $\lambda\geq\lambda^*$, our decision
algorithm will find a critical set and determine the order in which the sensors
of the critical set will appear in a feasible solution for $\lambda$.
For the purpose of giving some intuition and
later showing the correctness of our algorithm, we first
explore some properties of a critical set.

Consider a critical set $S^c$. For each sensor $s\in S^c$, we call the
set of points on $B$ that are covered by $s$ but not covered by any
other sensor in $S^c$ the {\em exclusive coverage} of $s$.

\begin{observation}\label{obser:uniquecov}
The exclusive coverage of each sensor in a critical set $S^c$ is a continuous portion of
the barrier $B$.
\end{observation}
\begin{proof}
Assume to the contrary the exclusive
coverage of a sensor $s\in S^c$ is not a continuous portion of $B$. Then there
must be at least one sensor $s' \in S^c$ whose covering interval is
between two consecutive continuous portions of the exclusive coverage of $s$.
But that would mean $s'$ is not critical since the covering interval of
$s'$ is contained in that of $s$. Hence, the observation holds.
\end{proof}

For a critical set $S^c$ in a feasible solution $SOL$,
we define the {\em cover order} of the
sensors in $S^c$ as the order of these sensors in $SOL$ such that
their exclusive coverages are from left to right.

\begin{observation}\label{obser:standardorder}
The cover order of the sensors of a critical set $S^c$ in a feasible solution
$SOL$ is consistent with the left-to-right order of the positions of these
sensors in $SOL$. Further, the cover order is also consistent with
the order of the {\em right} (resp., left) extensions of these sensors
in $SOL$.
\end{observation}
\begin{proof}
Consider any two sensors $s_i$ and $s_j$ in $S^c$ with ranges $r_i$
and $r_j$, respectively. Without loss of
generality, assume $s_i$ is to the left of $s_j$ in the cover order,
i.e., the exclusive coverage of $s_i$ is to the left of that of $s_j$ in
$SOL$.
Let $y_i$ and $y_j$ be the positions of $s_i$ and $s_j$ in $SOL$,
respectively. To prove the observation, it suffices to show
$y_i<y_j$, $y_i+r_i<y_j+r_j$, and $y_i-r_i<y_j-r_j$.

Let $p$ be a point in the exclusive coverage of $s_j$. We also use $p$ to
denote its coordinate on the $x$-axis.  Then $p$ is not
covered by $s_i$, implying either $p>y_i+r_i$ or $p<y_i-r_i$. But,
the latter case cannot hold (otherwise, the exclusive coverage of $s_i$
would be to the right of that of $s_j$). Since $p$ is covered by $s_j$,
we have $p\leq y_j+r_j$. Therefore, $y_i+r_i<p\leq y_j+r_j$. By using
a symmetric argument, we can also prove $y_i-r_i<y_j-r_j$
(we omit the details). Clearly, the two inequalities $y_i+r_i<y_j+r_j$
and $y_i-r_i<y_j-r_j$ imply $y_i<y_j$.
The observation thus holds.
\end{proof}

An interval $I$ of $B$ is called a {\em left-aligned interval} if the left
endpoint of $I$ is at $0$ (i.e., $I$ is of the form $[0,x']$ or
$[0,x')$). A set of sensors is said to be in {\em attached positions}
if the union of their covering intervals is a
continuous interval of the $x$-axis whose length is equal to the sum
of the lengths of these covering intervals. Two intervals of the
$x$-axis {\em overlap} if they intersect each other (even at only one
point).

\subsubsection{The Algorithm Description}
\label{subsec:desc}

Initially, we move all sensors of $S$ to the right by the distance
$\lambda$, i.e., for each $1\leq i\leq n$, we move $s_i$ to the position
$x_i'=x_i+\lambda$. Let $C_0$ denote the resulting configuration.
Clearly, there is a feasible solution for $\lambda$ if and only if
we can move the sensors in $C_0$ to the left by at most
$2\lambda$ to form a coverage of $B$. Thus, henceforth we only
need to consider moving the sensors to the left. Recall that we have
assumed that the extensions of any two distinct sensors are different; hence in
$C_0$, the extensions of all sensors are also different.

Our algorithm takes a greedy approach. It seeks
to find sensors to cover $B$ from left to right, in
at most $n$ steps. If $\lambda\geq\lambda^*$, the
algorithm will end up with a critical set $S^c$ of sensors along with the
destinations for all these sensors. In theory, the other sensors in $S\setminus
S^c$ can be anywhere such that their displacements are at most
$\lambda$; but in the solution found by our algorithm, they are
at the same positions as in $C_0$. If a sensor is at the same
position as in $C_0$, we say it {\em stands still}.

In step $i$ (initially, $i=1$), using the configuration $C_{i-1}$
produced in step $i-1$ and based on certain criteria, we find
a sensor $s_{g(i)}$ and determine its destination $y_{g(i)}$, where
$g(i)$ is the index of the sensor in $S$ and
$y_{g(i)}\in [x'_{g(i)}-2\lambda,x'_{g(i)}]$.
We then move the sensor $s_{g(i)}$ to $y_{g(i)}$ to obtain a new configuration
$C_i$ from $C_{i-1}$ (if $y_{g(i)}=x'_{g(i)}$, then we need not move
$s_{g(i)}$, and $C_i$ is the same as $C_{i-1}$).
Let $R_i=y_{g(i)}+r_{g(i)}$ (i.e., the right extension of $s_{g(i)}$ in $C_i$).
Assume $R_0=0$.
Let $S_i=S_{i-1}\cup\{s_{g(i)}\}$ ($S_0=\emptyset$ initially).
We will show that the sensors in $S_i$ together cover the left-aligned
interval $[0,R_i]$.
If $R_i\geq L$, we have found a feasible solution with a critical
set $S^c=S_i$, and terminate the algorithm. Otherwise, we proceed to
step $i+1$.  Further, it is possible that a desired sensor
$s_{g(i)}$ cannot be found, in which case we
terminate the algorithm and report $\lambda<\lambda^*$.
Below we give the details, and in particular, discuss
how to determine the sensor $s_{g(i)}$ in each
step.

Before discussing the first step, we provide some intuition.
Let $S_l$ consist of the sensors whose right extensions are at most
$0$ in $C_0$. We claim that since $L>0$, no sensor in
$S_l$ can be in a critical set of a feasible solution if
$\lambda^*\leq \lambda$.
Indeed, because all sensors have been moved to their rightmost possible
positions in $C_0$, if no sensor in $S_l$ has a right
extension at $0$ in $C_0$, then the claim trivially holds;
otherwise, suppose $s_t$ is such a sensor. Assume to the contrary that
$s_t$ is in a critical set $S^c$. Then $p(0)$ is the only point on $B$ that
can be covered by $s_t$. Since $L>0$, there must be another sensor in $S^c$
that also covers $p(0)$ (otherwise, no sensor in $S^c$ would cover
the point $p^+(0)$). Hence, $s_t$ is not critical with respect to $S^c$,
a contradiction. The claim thus follows. Therefore, we
need not consider the sensors in $S_l$ since they do not
help in forming a feasible solution.


In step 1, we determine the sensor $s_{g(1)}$, as follows.
Define $S_{11}=\{s_j\ |\ x_j'-r_j\leq 0 < x_j'+r_j\}$
(see Fig.~\ref{fig:case1}), i.e., $S_{11}$ consists of all sensors
covering the point $p(0)$ in $C_0$ except any sensor whose right extension is
$0$ (but if the left extension of a sensor is $0$, the sensor is included in
$S_{11}$). In other words, $S_{11}$ consists of all sensors covering the point
$p^+(0)$ in $C_0$.
If $S_{11}\neq \emptyset$, then we choose the sensor in $S_{11}$ whose right
extension is the largest as $s_{g(1)}$ (e.g., $s_i$ in Fig.~\ref{fig:case1}),
and let $y_{g(1)}=x'_{g(1)}$. Note that since
the extensions of all sensors in $C_0$ are different, the sensor
$s_{g(1)}$ is unique.
If $S_{11}= \emptyset$, then define $S_{12}$ as the set of sensors whose
left extensions are larger than $0$ and at most $2\lambda$ (e.g., see
Fig.~\ref{fig:case2}).
If $S_{12}=\emptyset$, then we terminate the algorithm and
report $\lambda<\lambda^*$. Otherwise, we choose the sensor in
$S_{12}$ whose right extension is the smallest as $s_{g(1)}$ (e.g., $s_i$
in Fig.~\ref{fig:case2}), and let $y_{g(1)}=r_{g(1)}$
(i.e., the left extension of $s_{g(1)}$ is at $0$ after it is moved to the
destination $y_{g(1)}$).

\begin{figure}[t]
\begin{minipage}[t]{0.49\linewidth}
\begin{center}
\includegraphics[totalheight=0.8in]{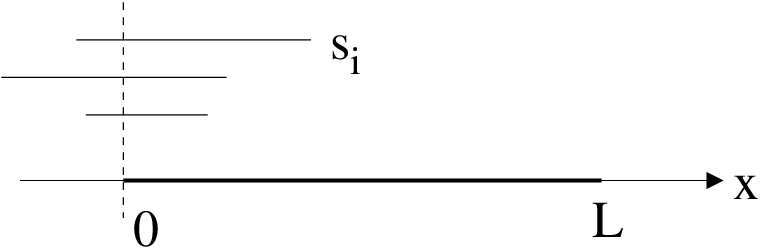}
\caption{\footnotesize The set $S_{11}$ consists of the three sensors
whose covering intervals are shown, and $s_{g(1)}$ is $s_i$.}
\label{fig:case1}
\end{center}
\end{minipage}
\hspace{0.02in}
\begin{minipage}[t]{0.49\linewidth}
\begin{center}
\includegraphics[totalheight=0.8in]{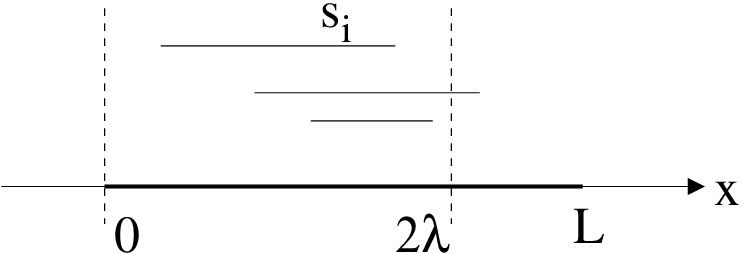}
\caption{\footnotesize The set $S_{12}$ consists of the three sensors
whose covering intervals are shown, and
$s_{g(1)}$ is $s_i$ if $S_{11}=\emptyset$.}
\label{fig:case2}
\end{center}
\end{minipage}
\vspace*{-0.15in}
\end{figure}

If the algorithm is not terminated, then we move $s_{g(1)}$
to $y_{g(1)}$, yielding a new configuration $C_1$. Let
$S_1=\{s_{g(1)}\}$, and $R_1$ be the
right extension of $s_{g(1)}$ in $C_1$. If $R_1\geq L$, we have found
a feasible solution $C_1$ with the critical set $S_1$, and terminate
the algorithm. Otherwise, we proceed to step 2.


The general step is very similar to step 1. Consider step $i$ for $i>1$.
We determine the sensor $s_{g(i)}$, as follows.
Let $S_{i1}$ be the set of sensors covering the point
$p^+(R_{i-1})$ in the configuration $C_{i-1}$.
If $S_{i1}\neq \emptyset$, we choose the sensor in $S_{i1}$ with the
largest right extension as $s_{g(i)}$ and let $y_{g(i)}=x'_{g(i)}$.
Otherwise, let $S_{i2}$ be the set of sensors whose
left extensions are larger than $R_{i-1}$ and at most $R_{i-1}+2\lambda$.
If $S_{i2}=\emptyset$, we terminate the algorithm and
report $\lambda<\lambda^*$. Otherwise, we choose the sensor in
$S_{i2}$ with the smallest right extension as $s_{g(i)}$ and let
$y_{g(i)}=R_{i-1}+r_{g(i)}$.
If the algorithm is not terminated, we move $s_{g(i)}$
to $y_{g(i)}$ and obtain a new configuration $C_i$. Let
$S_i=S_{i-1}\cup\{s_{g(i)}\}$. Let $R_i$ be the
right extension of $s_{g(i)}$ in $C_i$. If $R_i\geq L$, we have found
a feasible solution $C_i$ with the critical set $S_i$ and terminate
the algorithm. Otherwise, we proceed to step $i+1$.
If the sensor $s_{g(i)}$ is from $S_{i1}$ (resp., $S_{i2}$), then we call
it the {\em type I} (resp., {\em type II}) sensor.

Since there are $n$ sensors in $S$, the algorithm is terminated
in at most $n$ steps. This finishes the description of our algorithm.

\subsubsection{The Correctness of the Algorithm}
\label{subsec:correct}

Based on the description of our algorithm, we have the following lemma.

\begin{lemma}\label{lem:correct10}
At the end of step $i$, suppose the algorithm produces the set $S_i$ and
the configuration $C_i$; then $S_i$ and $C_i$ have the following properties.
\begin{description}
\item[(a)]
$S_i$ consists of sensors that are type I or type II.
\item[(b)]
For each sensor $s_{g(j)}\in S_i$ with $1\leq j\leq i$, if $s_{g(j)}$ is
of type I, then it stands still (i.e., its position in $C_i$ is the same as
that in $C_0$); otherwise, its left extension is at $R_{j-1}$,
and $s_{g(j)}$ and $s_{g(j-1)}$ are in attached positions if $j>1$.
\item[(c)]
The interval on $B$ covered by the
sensors in $S_i$ is $[0,R_i]$.
\item[(d)]
For each $1<j\leq i$, the right extension of $s_{g(j)}$ is larger than
that of $s_{g(j-1)}$.
\item[(e)]
For each $1\leq j\leq i$, $s_{g(j)}$ is the only sensor in $S_i$ that
covers the point $p^+(R_{j-1})$ (with $R_0=0$).
\end{description}
\end{lemma}
\begin{proof}
The first three properties are trivially true according to the algorithm description.

For property (d), note that the right extension of $s_{g(j)}$ (resp.,
$s_{g(j-1)}$) is $R_j$ (resp., $R_{j-1}$). According to our algorithm,
the sensor $s_{g(j)}$ covers the point $p^+(R_{j-1})$, implying that
$R_j>R_{j-1}$. Hence, property (d) holds.

For property (e), note that the sensor $s_{g(j)}$ (which is determined
in step $j$) always covers $p^+(R_{j-1})$.
Consider any other sensor $s_{g(t)}\in S_i$. If $t<j$, then the
right extension of $s_{g(t)}$ is at most $R_{j-1}$, and thus $s_{g(t)}$
cannot cover $p^+(R_{j-1})$. If $t>j$, then depending on whether
$s_{g(t)}\in S_{t1}$ or $s_{g(t)}\in S_{t2}$, there are two cases. If
$s_{g(t)}\in S_{t2}$, then the left extension of $s_{g(t)}$ is
$R_{t-1}$, which is larger than $R_{j-1}$, and thus $s_{g(t)}$ cannot
cover $p^+(R_{j-1})$ in $C_i$. Otherwise (i.e., $s_{g(t)}\in S_{t1}$),
$s_{g(t)}$ stands still.
Assume to the contrary that $s_{g(t)}$ covers $p^+(R_{j-1})$ in $C_i$.
Then $s_{g(t)}$ must have been in $S_{j1}$ in step $j$ within the
configuration $C_{j-1}$. This implies $S_{j1}\neq \emptyset$,
$s_{g(j)}\in S_{j1}$, and $s_{g(j)}$ stands still.
Since $R_t$ is the right extension of $s_{g(t)}$ and $R_j$ is the right
extension of $s_{g(j)}$, by property (d), for $t>j$, we have $R_t>R_j$.
Since $R_t>R_j$ (i.e., the right extension of $s_{g(j)}$ is smaller than that
of $s_{g(t)}$), the algorithm cannot choose $s_{g(j)}$ from $S_{j1}$ in
step $j$, which is a contradiction. Therefore, $s_{g(t)}$ cannot
cover the point $p^+(R_{j-1})$. Property (e) thus holds.
\end{proof}

At its termination, our algorithm either reports
$\lambda\geq\lambda^*$ or $\lambda<\lambda^*$.
To argue the correctness of the algorithm, below we will show that
if the algorithm reports $\lambda\geq\lambda^*$, then indeed there is
a feasible solution for $\lambda$ and our algorithm finds one;
otherwise, there is no feasible solution for $\lambda$.

Suppose in step $i$, our algorithm reports $\lambda\geq
\lambda^*$. Then according to the algorithm, it must be $R_i\geq L$. By Lemma
\ref{lem:correct10}(c) and
\ref{lem:correct10}(e), $C_i$ is a feasible solution and
$S_i$ is a critical set. Further, by Lemma \ref{lem:correct10}(d) and
Observation \ref{obser:standardorder}, the cover order of the
sensors in $S_i$ is $s_{g(1)},s_{g(2)},\ldots,s_{g(i)}$.

Next, we show that if the algorithm reports $\lambda<\lambda^*$, then
certainly there is no feasible solution for $\lambda$. This is almost an
immediate consequence of the following lemma.

\begin{lemma}\label{lem:correct20}
Suppose $S_i'$ is the set of sensors in the configuration $C_i$
whose right extensions are at most $R_i$.
Then the interval $[0,R_i]$ is the largest possible left-aligned interval
that can be covered by the sensors of $S'_i$ such that the
displacement of each sensor in $S'_i$ is at most $\lambda$.
\end{lemma}
\begin{proof}
In this proof, when we say an interval is covered by the sensors of
$S_i'$, we mean (without explicitly stating) that the displacement of
each sensor in $S_i'$ is at most $\lambda$.

We first prove a key claim: If $C$ is a configuration for
the sensors of $S'_i$ such that a left-aligned interval $[0,x']$ is
covered by the sensors of $S'_i$, then there always exists a
configuration $C^*$ for $S'_i$ in which the interval $[0,x']$ is still
covered by the sensors of $S'_i$ and for each $1\leq j\leq i$, the
position of the sensor $s_{g(j)}$ in $C^*$ is $y_{g(j)}$, where
$g(j)$ and $y_{g(j)}$ are the values computed by our algorithm.

As similar to our discussion in Section \ref{subsec:preli}, the
configuration $C$ for $S'_i$ always has a critical set for covering the
interval $[0,x']$. Let $S_C$ be such a critical set of $C$.

We prove the claim by induction. We first show the base case:
Suppose there is a configuration $C$ for the sensors of $S'_i$ in
which a left-aligned interval $[0,x']$ is covered by the sensors of
$S'_i$; then there is a configuration $C_1'$ for the sensors of
$S'_i$ in which the interval $[0,x']$ is still
covered by the sensors of $S'_i$ and the
position of the sensor $s_{g(1)}$ in $C_1'$ is $y_{g(1)}$.

Let $t=g(1)$. If the position of $s_t$ in $C$ is $y_t$, then we are
done (with $C_1'= C$). Otherwise, let $y_t'$ be the
position of $s_t$ in $C$, with $y_t'\neq y_t$.
Depending on $s_t\in S_{11}$ or $s_t\in S_{12}$, there are two cases.

If $s_t\in S_{11}$, then $y_t=x'_t$. Since $y_t$ is the rightmost position
to which the sensor $s_t$ is allowed to move and $y_t'\neq y_t$, we have
$y_t'<y_t$. Depending on whether $s_t$ is in the critical
set $S_C$, there further are two subcases.

If $s_t\not\in S_C$, then by the definition of a critical
set, the sensors in $S_C$ form a coverage of $[0,x']$
regardless of where $s_t$ is. If we move $s_t$ to
$y_t$ (and other sensors keep the same positions as in $C$) to
obtain a new configuration $C_1'$, then the sensors of $S_i'$ still
form a coverage of $[0,x']$.

If $s_t\in S_C$, then because $y_t>y_t'$, if we move $s_t$ from $y_t'$ to
$y_t$, $s_t$ is moved to the right. Since $s_t\in S_{11}$, when $s_t$
is at $y_t$, $s_t$ still covers the point $p(0)$. Thus, moving
$s_t$ from $y_t'$ to $y_t$ does not cause $s_t$ to cover a smaller
sub-interval of $[0,x']$.
Hence, by moving $s_t$ to $y_t$, we obtain a new configuration $C_1'$ in
which the sensors of $S_i'$ still form a coverage of $[0,x']$.

This completes the proof for the case $s_t\in S_{11}$.

If $s_t\in S_{12}$, then according to our algorithm, $S_{11}=\emptyset$
in this case, and $s_t$ is the
sensor in $S_{12}$ with the smallest right extension in $C_0$. If $s_t\not\in
S_C$, then by the same argument as above, we can obtain a
configuration $C_1'$ in which the interval $[0,x']$ is still
covered by the sensors of $S'_i$ and the
position of the sensor $s_t$ in $C_1'$ is $y_t$. Below, we discuss the
case when $s_t\in S_C$.

In $S_C$, some sensors must cover the point
$p(0)$ in $C$. Let $S'$ be the set of sensors in $S_C$ that cover
$p(0)$ in $C$. If $s_t\in S'$, then it is easy to see that $y_t'<y_t$
since $y_t$ is the rightmost position for $s_t$ to cover $p(0)$. In this
case, again, by the same argument as above, we can always move $s_t$ to the
right from $y_t'$ to $y_t$ to obtain a
configuration $C_1'$ in which the interval $[0,x']$ is still
covered by the sensors of $S'_i$. Otherwise (i.e., $s_t\not\in S'$),
we show below that we can always move $s_t$ to $y_t$ by switching the
relative positions of $s_t$ and some other sensors in $S_C$.

An easy observation is that each sensor in $S'$ must be in $S_{12}$.
Consider an arbitrary sensor $s_h\in S'$.  Since $s_t$ is the
sensor in $S_{12}$ with the smallest right extension in $C_0$, the
right extension of $s_h$ is larger than that of $s_t$ in $C_0$.
Depending on whether the covering intervals of
$s_t$ and $s_h$ overlap in $C$, there are two subcases.

If the covering intervals of
$s_t$ and $s_h$ overlap in $C$, then let $[0,x'']$ be the
left-aligned interval that is covered by $s_t$ and $s_h$ in $C$ (see
Fig.~\ref{fig:overlap}).
If we {\em switch} their relative positions by moving $s_t$ to $y_t$ and moving
$y_h$ to $x''-r_h$ (i.e., the left extension of $s_t$ is at $0$ and
the right extension of $s_h$ is at $x''$),
then these two sensors still cover $[0,x'']$ (see
Fig.~\ref{fig:overlap}), and thus the sensors in $S_i'$ still form a coverage
of $[0,x'']$.
Further, after the above switch operation, the displacements of these two
sensors are no bigger than $\lambda$. To see this, clearly, the
displacement of $s_t$ is at most $\lambda$.
For the sensor $s_h$,
it is easy to see that the switch operation moves $s_h$ to the right.
Since $s_t$ covers $p(x'')$ in $C$, $x''$ is no larger than the right extension of
$s_t$ in $C_0$, which is smaller than that of $s_h$ in $C_0$ since
$s_h$ is in $S_{12}$. Thus,
$x''$ is smaller than $x_h'+r_h$, implying that the position of $s_h$ after the
switch operation is still to the left of its position in $C_0$.
Hence, after the switch operation,
the displacement of $s_h$ is no bigger than $\lambda$. In
summary, after the switch operation, we obtain a new
configuration $C_1'$ in which the interval $[0,x']$ is still
covered by the sensors of $S'_i$ and the
position of the sensor $s_t$ in $C_1'$ is $y_t$.

\begin{figure}[t]
\begin{minipage}[t]{\linewidth}
\begin{center}
\includegraphics[totalheight=0.8in]{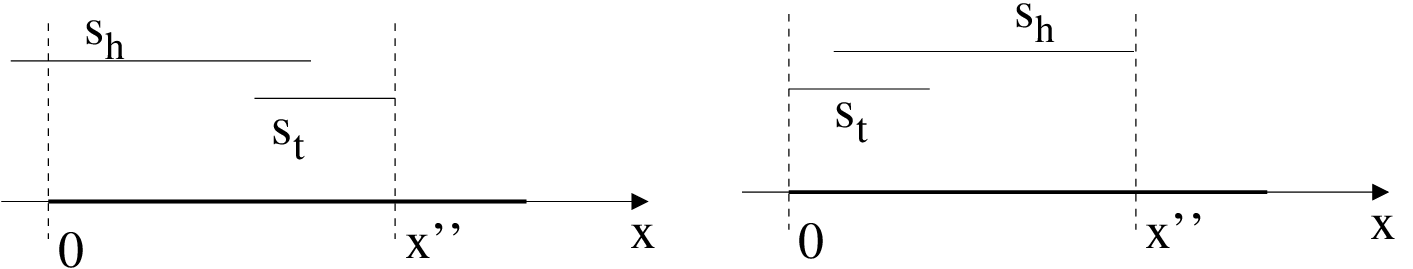}
\caption{\footnotesize Illustrating the switch operation on $s_t$ and
$s_h$: The left figure is before the switch and the right one is after the switch.}
\label{fig:overlap}
\end{center}
\end{minipage}
\vspace*{-0.15in}
\end{figure}

If the covering intervals of $s_t$ and $s_h$ do not overlap in $C$, then suppose
the sensors in
the critical set $S_C$ between $s_h$ and $s_t$ are
$s_h=s_{f(0)},s_{f(1)},s_{f(2)},\ldots,s_{f(m)},s_t$, in the cover order.
Clearly, the covering intervals of any two consecutive sensors in this list
overlap in $C$.
In the following, we perform a sequence of ``switch operations''
to obtain a configuration $C_1'$ in which $[0,x']$ is covered by the sensors
of $S'_i$ and the location of the sensor $s_t$ is $y_t$.
Let $Q=\{s_{f(0)},s_{f(1)},s_{f(2)},\ldots,s_{f(m)}\}$.

Indeed, this is the most difficult case\footnote{The proof for this case in the previous
version of this paper \cite{ref:ChenAl13} is based on the assumption that every
sensor of $Q$ is in $S_{12}$. However, the assumption is not true because it is
possible that some sensors of $Q$ are not in $S_{12}$. Therefore, the previous
proof is not correct. Here, we provide a new proof for this case.}.

For any sensor $s_k$ of $S$ and any configuration $C'$, we use $x_k(C')$ to denote its location in $C'$. Recall that $x_k'$ is the location of $s_k$ in $C_0$.
We use $I(s_k)$ to denote the covering interval of $s_k$. We say
that $s_k$ is {\em valid} in a configuration $C'$ if the displacement of $s_k$ in
$C'$ is at most $\lambda$ (i.e., $x_k(C')\in [x_k'-2\lambda,x_k']$). We say that a configuration $C'$ is {\em valid} if every sensor of $S$ is valid in $C'$.

Consider any sensor $s_{f(j)}$ in $Q$. If the right extension of $s_{f(j)}$ is strictly less than the right extension of $s_t$ in $C_0$, then we say $s_{f(j)}$ is a {\em special sensor}; otherwise it is a {\em normal} sensor (see Fig~\ref{fig:normalspecial}). Note that all sensors of $S_{12}$ are normal
sensors but not every normal sensor is in $S_{12}$. If a normal sensor
is not in $S_{12}$, then its left extension must be larger than $2\lambda$ in $C_0$.

\begin{figure}[t]
\begin{minipage}[t]{\linewidth}
\begin{center}
\includegraphics[totalheight=0.8in]{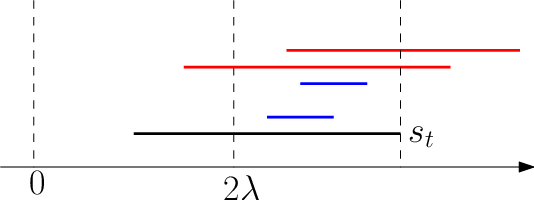}
\caption{\footnotesize Illustrating the special and normal sensors of $Q$ when they are in $C_0$: each segment represents the covering interval of a sensor of $Q$ in $C_0$; the black segment is for $s_t$; the red segments are for normal sensors and the blue segments are for special sensors.}
\label{fig:normalspecial}
\end{center}
\end{minipage}
\vspace*{-0.15in}
\end{figure}

We begin with considering the last sensor $s_{f(m)}$ of the sequence
$Q$.  Depending on whether $s_{f(m)}$ is a normal sensor or a
special sensor, there are two cases.


\paragraph{The normal sensor case.}
If $s_{f(m)}$ is a normal sensor, then let $x''_1$ be the maximum of $0$ and the
left extension of $s_{f(m)}$ in $C$, and $x''_2$ be the
right extension of $s_t$ in $C$ (see
Fig.~\ref{fig:nonoverlap}). Clearly, $x''_1<x''_2$.
We perform a switch operation on $s_t$ and $s_{f(m)}$ by moving $s_t$ to the left and
moving $s_{f(m)}$ to the right such that the left extension of $s_t$ is
at $x_1''$ and the right extension of $s_{f(m)}$ is at $x_2''$ (see
Fig.~\ref{fig:nonoverlap}). It is easy to see that after this switch operation,
the sensors in $S_C$ still form a coverage of $[0,x']$.

Let $C'$ be the new configuration after the above switch operation. Next we show that $C'$ is still valid. To this end, it is sufficient to show both $s_{f(m)}$ and $s_t$ are valid in $C'$.

\begin{figure}[h]
\begin{minipage}[t]{\linewidth}
\begin{center}
\includegraphics[totalheight=0.8in]{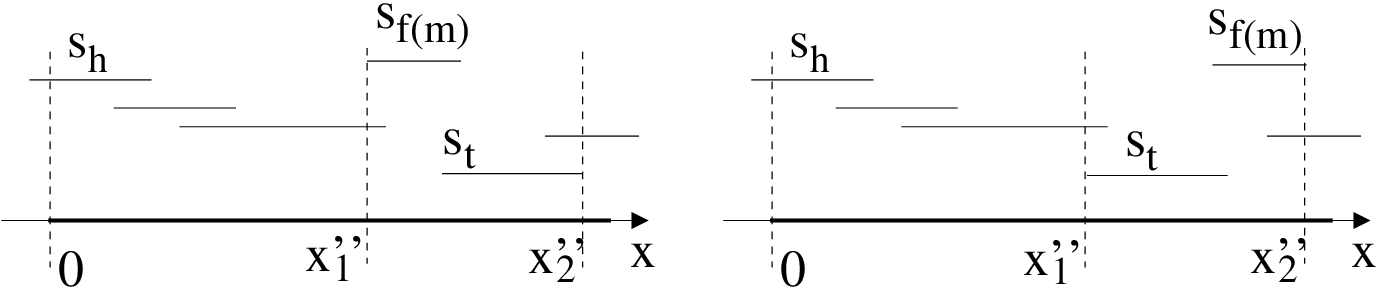}
\caption{\footnotesize Illustrating the switch between $s_t$ and $s_{f(m)}$.}
\label{fig:nonoverlap}
\end{center}
\end{minipage}
\end{figure}

We first consider $s_t$. Our goal is to show that $x_t(C')\in [x_t'-2\lambda,x_t']$.
On the one hand, since the left extension of $s_t$ in $C'$ is at $x_1''\geq 0$, it holds that $x_t(C')\geq y_t\geq x_t'-2\lambda$. On the other hand, since $s_t$ is valid in $C$ and $s_t$ has been moved leftwards from $x_t(C)$ to $x_t(C')$, we obtain that $x_t(C')\leq x_t(C)\leq x_t'$. This proves that $s_t$ is valid in $C'$.

We then consider $s_{f(m)}$. Let $k=f(m)$. Our goal is to show that
$x_{k}(C')\in [x_{k}'-2\lambda,x_{k}']$. On the one hand, since $s_k$
is valid in $C$ and $s_k$ has been moved rightwards from $x_k(C)$ to
$x_k(C')$, we obtain that $x_k'-2\lambda\leq x_k(C)\leq x_k(C')$. On
the other hand, since $s_{k}$ is a normal sensor, the right extension
of $s_{k}$ is larger than that of $s_t$ in $C_0$. However, the right extension
of $s_k$ in $C'$ is at $x_2''$, which is the right extension of $s_t$
in $C$. This implies that comparing with its location $x_k'$ in $C_0$,
$s_k$ must have been moved leftwards to reach $x_k(C')$. In other
words, $x_k(C')\leq x_k'$.
This proves that $x_{k}(C')\in [x_{k}'-2\lambda,x_{k}']$, and thus, $s_k$ is valid
in $C'$.

Thus, $C'$ is a valid configuration in which $[0,x']$ is still
covered by the sensors of $S_C$. But the sensors between $s_h$ and
$s_t$ in $C'$ are $s_{f(1)},s_{f(2)},\ldots,s_{f(m-1)}$, which has one
sensor less than before. Let
$Q=\{s_{f(0)},s_{f(1)},s_{f(2)},\ldots,s_{f(m-1)}\}$.
We proceed to apply the same argument on $Q$ and $C'$.


\paragraph{The special sensor case.}
If $s_{f(m)}$ is a special sensor, this case is more complicated.
Let $s_{f(j)}$ be the last normal sensor
in the sequence $Q$ for some $j$ with $0\leq j\leq m-1$. Note
that such a normal sensor must exist in $Q$ because $s_{f(0)}=s_h$
is a normal sensor. Let $\calI=[\alpha,\beta]$ be the union of the covering
intervals of $s_{f(j+1)},s_{f(j+2)},\ldots,s_{f(m)}$ in $C$. Let $\calI'=[\alpha',\beta']$
be the union of $\calI$ and the covering intervals of $s_{f(j)}$ and $s_t$ in $C$.
For ease of exposition, we assume $\calI'\subseteq [0,x']$ (since
otherwise we could apply similar argument on the interval
$[0,\infty)\cap I'$, as we did for the above normal sensor case). Note
that since all sensors of $Q\cup \{s_{f(j)},s_t\}$ are in the critical
set $S_C$ and their order follows the cover order,
the left extension of $s_{f(j)}$ is at $\alpha'$ and the
right extension of $s_t$ is at $\beta'$ in $C$.

Consider any sensor $s_{f(k)}$ with $j+1\leq k\leq m$. Since
$s_{f(k)}$ is a special sensor, we claim that the left extension of
$s_{f(k)}$ in $C_0$ is larger than $2\lambda$ (see
Fig.~\ref{fig:normalspecial}). Indeed, assume to the contrary that
this is not true. Then, since $S_{11}=\emptyset$ and $s_{f(k)}\not\in
S_{12}$, the right extension of $s_{f(k)}$ in $C_0$ must be at most
$0$. However, this implies that $s_{f(k)}$ cannot be in the critical
set $S_C$ for covering the interval $[0,x']$, incurring contradiction.
This proves the claim.


Let $Q'=\{s_{f(j+1)},s_{f(j+2)},\ldots,s_{f(m)}\}$. Define $Q'_1$ to
be the
set of sensors of $Q'$ whose right extensions are less than the left
extension of $s_{f(j)}$ in $C_0$. Let $Q_2'=Q'\setminus Q_1'$
(see Fig.~\ref{fig:special}).

\begin{figure}[h]
\begin{minipage}[t]{\linewidth}
\begin{center}
\includegraphics[totalheight=1.0in]{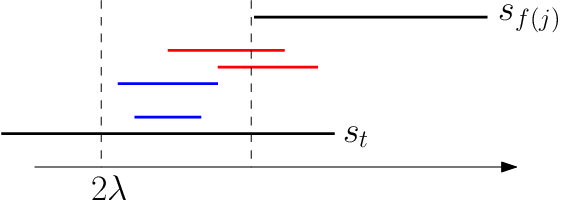}
\caption{\footnotesize Illustrating $Q_1'$ and $Q_2'$. The segments are covering intervals of sensors in $Q'\cup \{s_t,s_{f(j)}\}$ in $C_0$. The set $Q_1'$ (resp., $Q_2'$) consists of the two sensors whose covering intervals are blue (resp., red). }
\label{fig:special}
\end{center}
\end{minipage}
\end{figure}

Note that if $s_{f(j)}$ is in $S_{k2}$, then $Q'_1=\emptyset$. To see
this, suppose $s_{f(j)}$ is in $S_{k2}$. Then, the left extension of
$s_{f(j)}$ in $C_0$ is at most $2\lambda$. Since each sensor of $Q'$
is a special sensor, its right extension must be
larger than $2\lambda$. Thus, every sensor of $Q'$ has its right
extension larger than the left extension of $s_{f(j)}$ in $C_0$.
Hence,  $Q'_1=\emptyset$.

Let $Q_1'=\{s_{f_1(1)},s_{f_1(2)},\ldots,s_{f_1(m_1)}\}$ sorted by
their index order in $Q'$.


In the sequel, we will perform a sequence of operations to move sensors of $Q'\cup
\{s_{f(j)},s_t\}$ in $C$ to obtain a new configuration $C'$.
Roughly speaking, in $C'$, the order of the sensors of $Q'$ from left
to right are sensors of $Q'_1$, $s_t$,
sensors of $Q'_2$, and $s_{f(j)}$. We will show that after the operations, $\calI'$ is still covered by the sensors of $Q'\cup
\{s_{f(j)},s_t\}$ and $C'$ is a valid configuration.
In the following discussion, for each sensor $s_k\in Q'$, $x_k(C')$ refers to its location in $C'$ that we are going to set.

\begin{enumerate}
\item
We first consider sensors of $Q_1'$.

We consider the sensors of $Q_1'$ by their index order.
By the definition of $\calI'=[\alpha',\beta']$, the left extension of $s_{f_1(1)}$ in $C$ is larger than $\alpha'$.
We move  $s_{f_1(1)}$ leftwards until its left extension is at $\alpha'$.

Then, we consider $s_{f_1(2)}$. Clearly, either $I(s_{f_1(2)})$ intersects
$I(s_{f_1(1)})$ or $I(s_{f_1(2)})$ is to the right of $I(s_{f_1(1)})$. If the
former case happens, we do not move $s_{f_1(2)}$; otherwise, we move
$s_{f_1(2)}$ leftwards until its left extension is at the right extension of
$s_{f_1(1)}$.

We continue to consider other sensors of $Q_1'$. In general, suppose sensor
$s_{f_1(k)}$ has been considered as above. Then, we consider sensor
$s_{f_1(k+1)}$. Either $I(s_{f_1(k+1)})$ intersects
$I(s_{f_1(k)})$ or $I(s_{f_1(k+1)})$ is to the right of $I(s_{f_1(k)})$. If the
former case happens, we do not move $s_{f_1(k+1)}$; otherwise, we move
$s_{f_1(k+1)}$ leftwards until its left extension is at the right extension of
$s_{f_1(k)}$.

The above procedure stops once all sensors of $Q_1'$ have been considered. We
have the following observations. First,
each sensor of $Q_1'$ either does not change its position
or has been moved leftwards. Second, sensors of $Q_1'$ in $C'$ together cover
a continuous interval, denoted by $I(Q_1')$, whose left endpoint is $\alpha'$.
Further, due to the above
first observation, it holds that $I(Q_1')\subseteq \calI'=[\alpha',\beta']$.
Third, note that sensors of $Q_1'$ may cover multiple disjoint maximal
sub-intervals of $\calI'$ in $C$. Our above way of moving sensors of $Q_1'$
leftwards guarantees that the length of $I(Q_1')$ is larger than or equal to the
sum of the lengths of the above sub-intervals of $\calI'$ covered by sensors of
$Q_1'$ in $C$.

In the following, we show that each sensor of $Q_1'$ is valid in $C'$.
Consider any sensor $s_k\in Q_1'$. It is sufficient to show that
$x_k(C')\in [x_k'-2\lambda,x_k']$.

Due to the above first observation, $x_k(C')\leq x_k(C)$. Since
$s_k$ is valid in $C$, $x_k(C)\leq x_k'$. Thus, we have
$x_k(C')\leq x_k'$. Next, we show that $x_k(C')\geq x_k'-2\lambda$.

On the one hand, since $s_k\in Q_1'$, the left
extension of $s_k$ is smaller than the left extension of $s_{f(j)}$ in
$C_0$. On the other hand, the left extension of $s_{f(j)}$ is at
$\alpha'$ in $C$ while the left extension of $s_k$ is at least
$\alpha'$ in $C'$. This implies the leftward moving distance of $s_k$
from $x_k'$ in $C_0$ to  $x_k(C')$ in $C'$ is smaller than the leftward moving distance of
$s_{f(j)}$ from $x_{f(j)}'$ in $C_0$ to $x_{f(j)}(C)$ in $C$. Formally, $x_k'-x_k(C')\leq
x_{f(j)}'-x_{f(j)}(C)$. Since $s_{f(j)}$ is valid in $C$, we have
$x_{f(j)}'-x_{f(j)}(C)\leq 2\lambda$. Hence, we obtain $x_k'(C)\geq x_k'-2\lambda$.

This proves that $s_k$ is valid in $C'$.


Let $\beta_1'$ denote the right endpoint of $I(Q_1')$. Thus,  $I(Q_1')=[\alpha',\beta_1']$.
\item
We consider the sensor $s_t$.

We move $s_t$ to a new location $x_t(C')$ as follows. Since the right extension
of $s_t$ is at $\beta'$ and $I(Q_1')\subseteq \calI'$,
either $I(s_t)$  intersects $I(Q_1')$ or $I(s_t)$ is to the right of
$I(Q_1')$.

If the former case happens, the interval $\calI'=[\alpha',\beta']$ is currently covered by
sensors in $Q_1'$ and $s_t$. Hence, we do not need to move the sensors of $Q_2'\cup
\{s_{f(j)}\}$ any more. We then proceed to apply the
same argument on the
sequence of sensors $s_{f(0)},s_{f(1)},\ldots,s_{f(j-1)}$ followed by sensors
in $Q_1'$ in the configuration $C'$. Note that $s_{f(j)}$ is not in
the above sequence of sensors.

Below, we assume $I(s_t)$ is to the right of $I(Q_1')$. We move $s_t$
leftwards until its left extension is at $\beta_1'$.
In this way, sensors in $Q_1'$ and $s_t$ together
cover the interval $[\alpha',\beta_1'+2r_t]$ in $C'$.
Note that $s_t$ is valid in $C'$.
The proof is similar to that in the normal sensor case and we omit it.


\item
We consider the sensor $s_{f(j)}$.

By the definition of $\calI'=[\alpha',\beta']$,
the right extension of $s_{f(j)}$ is to the left of $\beta'$.
We move $s_{f(j)}$ rightwards such that its right extension is at $\beta'$.
Below, we show that $s_{f(j)}$ is still valid in $C'$.

It is sufficient to show that $s_{f(j)}(C')\in [x_{f(j)}'-2\lambda,x'_{f(j)}]$.
First of all, since $s_{f(j)}$ is valid in $C$ and we have moved it rightwards
to $x_{f(j)}(C')$, it holds that $x_{f(j)}(C')\geq x_{f(j)}(C)
\geq x'_{f(j)}-2\lambda$. Next, we show that $x_{f(j)}(C')\leq x'_{f(j)}$.

On the one hand, since $s_{f(j)}$ is a normal sensor, its right extension is larger than
that of $s_t$ in $C_0$. On the other hand, the right extension of $s_{f(j)}$ in
$C'$ is at $\beta'$, which is the same as the right extension of $s_t$ in $C$.
This implies that comparing with its location $x_{f(j)}'$ in $C_0$,
$s_{f(j)}$ must have been moved leftwards to reach $x_{f(j)}(C')$ in
$C'$. In other words, it holds that $x_{f(j)}(C')\leq x_{f(j)}'$.


This shows that $s_{f(j)}$ is still valid in $C'$.

\item
Finally, we consider the sensors of $Q_2'$.

Recall that currently in $C'$, sensors of $Q_1'$ together cover the interval
$I(Q_1')=[\alpha',\beta_1']$,  $s_t$ covers $[\beta_1',\beta_1'+2r_t]$, and
sensor $s_{f(j)}$ covers $[\beta'-2r_{f(j)},\beta']$.

If $\beta_1'+2r_t\geq \beta'-2r_{f(j)}$, then $\calI'=[\alpha',\beta']$ is
covered by sensors of $Q_1'\cup\{s_t,s_{f(j)}\}$. In this cae, we do not need to move
any sensor of $Q_2'$.

In the following, we assume $\beta_1'+2r_t< \beta'-2r_{f(j)}$. We claim that the
length of the interval $[\beta_1'+2r_t,\beta'-2r_{f(j)}]$ is no larger than the
sum of the lengths of the covering intervals of all sensors in $Q_2'$. Indeed,
we have shown above the the length of $I(Q_1')$ is no smaller than the sum of
the sub-invervals of $\calI'$ covered by the sensors of $Q_1'$ in $C$. On the
other hand, in $C'$ the three intervals $I(Q_1')$, $I(s_t)$, and $I(s_{f(j)})$ are
pairwise disjoint execept at their endpoints. Since $\calI'$ are covered by all
sensors of $Q'_1\cup Q_2'\cup \{s_t,s_{f(j)}\}$ in $C$. The claim must hold.

In the sequel, we give a {\em moving algorithm} to move sensors of $Q_2'$
to cover the interval $[\beta_1'+2r_t,\beta'-2r_{f(j)}]$  such that every sensor
of $Q_2'$ is still valid.

Initially, let $I$ refer to the covering interval $I(s_{f(j)})$, and we use
$\alpha(I)$  to denote the left endpoint of $I$. Thus, $I=[\alpha(I),\beta']$.
During the algorithm we will update $I$ and $\alpha(I)$.
Let $I'=[\alpha',\beta_1'+2r_t]$. Note that initially $\beta_1'+2r_t$ is less
than $\alpha(I)$.
We take an arbitrary sensor $s_k$ from $Q_2'$ and remove it from $Q_2'$. We do
the following on $s_k$.

We move $s_k$ such that the right extension of $s_k$ is at $\alpha(I)$. In the
following, we show that $s_k$ is valid in $C'$ (i.e., after the above movement).
Our goal is to show that $x_k(C')\in [x_k'-2\lambda,x_k']$.

On the one hand, since $s_k$ is a special sensor, the right extension of $s_k$ is
smaller than the right extension of $s_t$ in $C_0$. In $C'$,
the right extension of $s_k$ (i.e., $\alpha(I)$)
is larger than the right extension of
$s_t$ (i.e., $\beta_1'+2r_t$).
This means that the left moving distance of $s_t$ from $x_t'$ in
$C_0$ to $x_t(C')$ in $C'$ is larger than the left moving distance of $s_k$ from
$x_k'$ in $C_0$ to $x_k(C')$ in $C'$. Since we have proved above that $s_t$ is
valid in $C'$, we conclude that the left moving distance of $s_k$ from
$x_k'$ in $C_0$ to $x_k(C')$ in $C'$ is no more than $2\lambda$, i.e.,
$x_k(C')\geq x_k'-2\lambda$.

On the other hand, since $s_k$ is from $Q_2'$, the right extension of $s_k$ is
larger than or equal to the left extension of $s_{f(j)}$ in $C_0$. However, in $C'$, the
right extension of $s_k$ is smaller than or equal to the left extension of
$s_{f(j)}$.  This means that comparing with its location $x_k'$ in $C_0$, $s_k$ must
have been moved leftwards to reach $x_{k}(C')$ in $C'$. In other words, it holds
that $x_k(C')\leq x_k'$.

This proves that $s_k$ is still valid in $C'$.

Next we check whether $I(s_k)$ now intersects $I'$. If yes, then
$\calI'\subseteq I'\cup I(s_k)\cup I$, which means that $\calI'$ is now covered
in $C'$. Hence, we stop the moving algorithm.
Otherwise, we update $I=I(s_k)\cup I$ and $\alpha(I)=\alpha(I)-2r_k$. We proceed
to take another sensor of $Q_2'$ and apply the same moving algorithm as
above.

Due to our previous claim that the length of the interval
$[\beta_1'+2r_t,\beta'-2r_{f(j)}]$ is no larger than the
sum of the lengths of the covering intervals of all sensors in $Q_2'$.
The above moving algorithm will eventually stop and $\calI'$ will be
covered in $C'$. The above also shows that every sensor is valid in
$C'$.

Next, we continue to apply the similar argument on the new sequence of
sensors
$s_{f(1)},s_{f(2)},\ldots,s_{f(j-1)}$ followed by sensors
in $Q_1'$ and $s_t$ in the new configuration $C'$.
Note that $s_{f(j)}$ is not in the above sequence.

\end{enumerate}

The above shows that in either case (i.e., $s_{f(m)}$ is a normal or
special sensor) we can always move $s_t$ leftwards to ``switch'' with
a normal sensor. We keep doing the same thing on $s_t$ and the
remaining sequence of sensors until $s_h$ is considered, in which case
$s_t$ will be moved to $y_t$ (note that $Q_1'=\emptyset$ when $s_h$ is
considered since $s_h$ is in $S_{12}$).

This completes the proof of the base case, i.e., there is always a
configuration $C_1'$ in which the interval $[0,x']$ is
covered by the sensors of $S'_i$ and the
position of the sensor $s_{g(1)}$ in $C_1'$ is $y_{g(1)}$.

We assume inductively that the claim holds for each $k-1$ with $2\leq k\leq i$,
i.e., there is a configuration $C_{k-1}'$ in which the interval $[0,x']$ is
covered by the sensors of $S'_i$ and the
position of the sensor $s_{g(j)}$ for each $1\leq j\leq k-1$ in
$C_{k-1}'$ is $y_{g(j)}$. In the following, we show that the claim
holds for $k$, i.e., there is a configuration $C_{k}'$ in which the
interval $[0,x']$ is still covered by the sensors of $S'_i$ and the
position of the sensor $s_{g(j)}$ for each $1\leq j\leq k$ in $C_{k}'$ is
$y_{g(j)}$. The proof is quite similar to that for the base case and
we only discuss it briefly below.

Let $t=g(k)$. If the position of $s_t$ in $C_{k-1}'$ is $y_t$, then we are
done (with $C_k'= C'_{k-1}$). Otherwise, let $y_t'$ be the
position of $s_t$ in $C_{k-1}'$, with $y_t'\neq y_t$.
Depending on $s_t\in S_{k1}$ or $s_t\in S_{k2}$, there are
two cases.

\begin{itemize}
\item
If $s_t\in S_{k1}$, then $y_t=x'_t$. Since $y_t$ is the rightmost position
to which $s_t$ is allowed to move and $y_t'\neq y_t$, we have
$y_t'<y_t$. Depending on whether $s_t$ is in the critical
set $S_C$, there further are two subcases.

If $s_t\not\in S_C$, then the sensors in $S_C$ always form a coverage of
$[0,x']$ regardless of where $s_t$ is. Thus, if we move $s_t$ to
$y_t$, we obtain a new configuration $C_k'$ from $C_{k-1}'$ in which the
sensors of $S_i'$ still form a coverage of $[0,x']$ and the
position of the sensor $s_{g(j)}$ for each $1\leq j\leq k$ in
$C_{k}'$ is $y_{g(j)}$.

If $s_t\in S_C$, then since $y_t>y_t'$, if we move $s_t$ from $y_t'$ to
$y_t$, $s_t$ is moved to the right. By Lemma \ref{lem:correct10}(c), the
interval $[0,R_{k-1}]$ is covered by the sensors of $S_{k-1}
=\{s_{g(1)},s_{g(2)},\ldots, s_{g(k-1)}\}$ in $C_{k-1}'$
(since they are in positions $y_{g(1)},y_{g(2)},\ldots, y_{g(k-1)}$,
respectively).  When $s_t$ is at $y_t$, $s_t$ still covers the point
$p^+(R_{k-1})$. Thus, after moving
$s_t$ to $y_t$, we obtain a new configuration $C_k'$ from $C_{k-1}'$ in
which the sensors of $S_i'$ still form a coverage of $[0,x']$.

\item
If $s_t\in S_{k2}$, then $S_{k1}=\emptyset$ in this case, and $s_t$ is
the sensor in $S_{k2}$ with the smallest right extension. If $s_t\not\in
S_C$, then by the same argument as above, we can obtain a
configuration $C_k'$ from $C_{k-1}'$ in which the interval $[0,x']$ is still
covered by the sensors of $S'_i$ and the
position of the sensor $s_{g(j)}$ for each $1\leq j\leq k$ in
$C_{k}'$ is $y_{g(j)}$. Below, we discuss the
case when $s_t\in S_C$.

In $S_C$, some sensors must cover the point
$p^+(R_{k-1})$ in $C$. Let $S'$ be the set of sensors in $S_C$ that cover
$p^+(R_{k-1})$ in $C$. If $s_t\in S'$, then $y_t'<y_t$
since $y_t$ is the rightmost position for $s_t$ to cover $p^+(R_{k-1})$. In this
case, again, by the same argument as above, we can move $s_t$ to the
right from $y_t'$ to $y_t$ to obtain a
configuration $C_k'$ from $C_{k-1}'$ in which the interval $[0,x']$ is still
covered by the sensors of $S'_i$.
Otherwise (i.e., $s_t\not\in S'$), consider a sensor $s_h$ in $S'$.
Let the sensors in $S_C$ between $s_h$ and $s_t$ in the cover order
be $s_h,s_{f(1)},s_{f(2)},\ldots,s_{f(m)},s_t$ (this sequence may
contain only $s_h$ and $s_t$). Note that for each $1\leq j\leq k-1$, the
sensor $s_{g(j)}$ is not in this sequence.
Then by using a similar sequence of switch operations as
for the base case, we can obtain a new configuration $C_k'$ from $C_{k-1}'$
such that the sensors of $S'_i$ still form a coverage of $[0,x']$.
Again, the position of the sensor $s_{g(j)}$ for each $1\leq j\leq k$ in
$C_{k}'$ is $y_{g(j)}$.

\end{itemize}

This proves that the claim holds for $k$. Therefore, the claim is
true. The lemma can then be easily proved by using this claim, as
follows.

Suppose the largest left-aligned interval that can be covered by the
sensors of $S_i'$ is $[0,x']$. Then by the above claim,
there always exists a
configuration $C^*$ for $S'_i$ in which the interval $[0,x']$ is
also covered by the sensors of $S'_i$ and for each $1\leq j\leq i$, the
position of the sensor $s_{g(j)}$ in $C^*$ is $y_{g(j)}$. Recall that
$S_i=\{s_{g(1)},s_{g(2)},\ldots,s_{g(i)}\}$.
Then for each sensor $s_t\in S_i'\setminus S_i$, the rightmost point
that can be covered by $s_t$ is $x'_t+r_t$. Recall that in the
configuration $C_i$, for each $1\leq j\leq i$, the
position of the sensor $s_{g(j)}$ is $y_{g(j)}$, and for each sensor
$s_t\in S_i'\setminus S_i$, the position of $s_t$ is $x_t'$. Further,
by the definition of $S_i'$, the right extensions of all sensors in
$S_i'$ are at most $R_i$ in $C_i$.
Therefore, the right extensions of all sensors in $S_i'$ are also at most $R_i$ in $C^*$, implying that $x'\leq R_i$. On the other hand, by Lemma
\ref{lem:correct10}(c), the sensors of $S_i$ form a coverage of
$[0,R_i]$ in $C^*$. Thus, $[0,x'] =[0,R_i]$, and the lemma follows.
\end{proof}

Finally, we prove the correctness of our algorithm based on Lemma \ref{lem:correct20}.
Suppose our algorithm reports $\lambda<\lambda^*$ in step $i$.  Then
according to the algorithm, $R_{i-1}<L$ and both $S_{i1}$ and $S_{i2}$ are
$\emptyset$. Let $S_{i-1}'$ be the set of sensors whose right
extensions are at most $R_{i-1}$ in $C_{i-1}$. Since both $S_{i1}$ and
$S_{i2}$ are $\emptyset$, no sensor in $S\setminus S'_{i-1}$ can
cover any point to the left of the point $p^+(R_{i-1})$ (and including
$p^+(R_{i-1})$). By Lemma \ref{lem:correct20}, $[0,R_{i-1}]$ is the largest
left-aligned interval that can be covered by the sensors of $S_{i-1}'$.
Hence, the sensors in $S$ cannot cover the interval
$[0,p^+(R_{i-1})]$. Due to $R_{i-1}<L$, we have
$[0,p^+(R_{i-1})]\subseteq [0,L]$; thus the sensors of $S$ cannot cover
$B=[0,L]$. In other words, there is no feasible solution
for the distance $\lambda$. This establishes the correctness of our algorithm.

\subsubsection{The Algorithm Implementation}
\label{subsec:imple}

For the implementation of the algorithm, we first discuss a straightforward approach that runs in $O(n\log n)$ time. Later, we give another approach which, after $O(n\log n)$ time preprocessing, can determine whether $\lambda^*\leq \lambda$ in $O(n)$ time for any given $\lambda$. Although the second approach does not change the overall running time of our decision algorithm, it does help our
optimization algorithm in Section \ref{sec:algogeneral} to run faster.

In the beginning, we sort the $2n$ extensions of all sensors by the
$x$-coordinate, and move each sensor $s_i\in S$ to $x_i'$ to
produce the initial configuration $C_0$. During the algorithm, for each step $i$, we maintain two sets of sensors, $S_{i1}$ and $S_{i2}$, as defined earlier.
To this end, we
sweep along the $x$-axis and maintain $S_{i1}$ and $S_{i2}$,
using two {\em sweeping points} $p_1$ and $p_2$, respectively.
Specifically, the point $p_1$ follows the
positions $R_0$ ($=0$), $R_1,R_2,\ldots$, and $p_2$ follows the positions
$R_0+2\lambda,R_1+2\lambda,R_2+2\lambda,\ldots$. Thus, $p_2$ is kept
always by a distance of $2\lambda$ to the right of $p_1$.  To maintain the
set $S_{i1}$, when the sweeping point $p_1$ encounters the left
extension of a sensor, we insert the sensor into $S_{i1}$;
when $p_1$ encounters the right extension of a sensor, we delete the sensor from
$S_{i1}$. In this way, when the sweeping point $p_1$ is at $R_{i-1}$, we have the
set $S_{i1}$ ready. To maintain $S_{i2}$, the situation is slightly
more subtle. First, whenever the sweeping point $p_2$ encounters the left
extension of a sensor, we insert the sensor into $S_{i2}$. The subtle part
is at the deletion operation. By the definition of $S_{i2}$, if the left
extension of any sensor is less than or equal to $R_{i-1}$, then it should
not be in $S_{i2}$. Since eventually the first sweeping point $p_1$ is at
$R_{i-1}$ in step $i$, whenever a sensor is inserted into
the first set $S_{i1}$, we need to delete that sensor from $S_{i2}$.
Thus, a deletion on $S_{i2}$ happens only when the same sensor is
inserted into $S_{i1}$.
In addition, we need a {\em search} operation on $S_{i1}$ for finding the
sensor in $S_{i1}$ with the largest right extension, and a {\em search}
operation on $S_{i2}$ for finding the sensor in $S_{i2}$ with the smallest right
extension.

It is easy to see that there are $O(n)$ insertions and deletions in
the entire algorithm. Further, the search operations on both $S_{i1}$
and $S_{i2}$ are dependent on the right extensions of the senors.
By using a balanced binary search tree to represent each of these two
sets in which the right extensions of the sensors are used as keys,
the algorithm runs in $O(n\log n)$ time.

In the sequel, we give the second approach which, after $O(n\log n)$ time
preprocessing, can determine whether $\lambda^*\leq \lambda$ in $O(n)$ time for any given $\lambda$.

In the preprocessing, we compute two sorted lists $S_L$ and $S_R$,
where $S_L$ contains all sensors sorted by the increasing values of
their left extensions and $S_R$ contains all sensors sorted by the
increasing values of their right extensions. Consider any value
$\lambda$. Later in the algorithm, for each step $i$, our algorithm
will determine
the sensor $s_{g(i)}$ by scanning the two lists. We will show that
when the algorithm finishes, each sensor in $S_L$ is scanned at most
once and each sensor in $S_R$ is scanned at most three times, and
therefore, the algorithm runs in $O(n)$ time.

Initially, we move each sensor $s_i\in S$ to $x_i'$ to
produce the initial configuration $C_0$.  During the algorithm, we
sweep along the $x$-axis, using a {\em sweeping point} $p_1$.
Specifically, the point $p_1$ follows the positions $R_0$ ($=0$), $R_1,R_2,\ldots$.
With a little abuse of notation, we also let $p_1$ be the coordinate of the current position of $p_1$. Initially, $p_1=0$.

Consider a general step $i$ and we need to determine the sensor
$s_{g(i)}$. In the beginning of this step, $p_1$ is at the position
$R_{i-1}$. We scan the list $S_L$ from the beginning until the left
extension of the next sensor is strictly to the right of $p_1$. For
each scanned sensor $s_j$, if its right extension is strictly to the
right of $p_1$, then it is in $S_{i1}$ by the definition of $S_{i1}$.
Thus, the above scanning procedure can determine $S_{i1}$, after which
we can easily find the sensor $s_{g(i)}$ in $S_{i1}$ if $S_{i1}\neq
\emptyset$. In fact, we can compute $s_{g(i)}$ directly in the above
scanning procedure. In addition, for any sensor in $S_L$ that is
scanned, we remove it from $S_L$ (and thus it will never be scanned
later in the algorithm any more). If $S_{i1}\neq \emptyset$, then
$s_{g(i)}$ is determined and we move the sweeping point $p_1$ to the
right extension of $s_{g(i)}$ (i.e., $p_1=R_i$). If $p_1\geq L$, we
terminate the algorithm and report $\lambda^*\leq \lambda$; otherwise,
we continue on to step $i+1$. Below, we discuss the case
$S_{i1}=\emptyset$.

\begin{figure}[t]
\begin{minipage}[t]{\linewidth}
\begin{center}
\includegraphics[totalheight=0.8in]{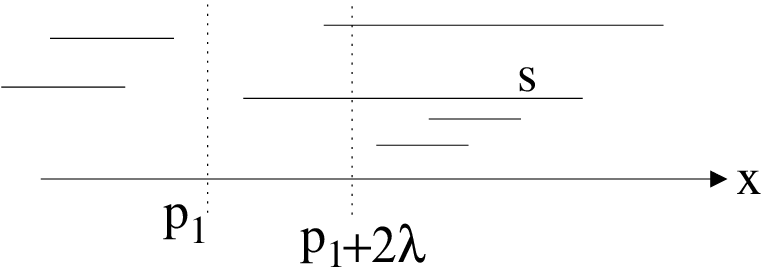}
\caption{\footnotesize Illustrating the search for the sensor $s$: the
sensors before $p_1$ are removed from $S_R$; the
two sensors below $s$ are redundant sensors for this step.}
\label{fig:redundant}
\end{center}
\end{minipage}
\vspace*{-0.15in}
\end{figure}

If $S_{i1}=\emptyset$, then $s_{g(i)}$ is the sensor in $S_{i2}$ with the smallest right extension if $S_{i2}\neq \emptyset$.
Specifically, among the sensors (if any) whose left extensions are
larger than $p_1$ ($=R_{i-1}$) and at most $p_1+2\lambda$, $s_{g(i)}$
is the sensor with the smallest right extension. To find $s_{g(i)}$,
we scan the list $S_R$ from the beginning until we find the first
sensor $s$ whose left extension is larger than $p_1$ and at most
$p_1+2\lambda$ (see Fig.~\ref{fig:redundant}). If such a sensor $s$ does not exist, then
$S_{i2}=\emptyset$, and we terminate the algorithm and report $\lambda^*>\lambda$.
Below, we assume we have found such a sensor $s$. Since the sensors in
$S_R$ are sorted by their right extensions, $s_{g(i)}$ is exactly the
sensor $s$. Further, unlike the scanning on $S_L$ where each scanned
sensor is removed immediately, for each scanned sensor in $S_R$, we
remove it only if its right extension is to the left of $p_1$
(see Fig.~\ref{fig:redundant}). Specifically, when we are searching the above sensor $s$ during scanning $S_R$, we remove from $S_R$ those sensors whose right extensions are to the left of $p_1$.
It is
easy to see that the removed sensors (if any) are consecutive from the
beginning of $S_R$. Let $S_R$ be the list after all removals. If
$s_{g(i)}$ ($=s$) is not the first sensor in $S_R$, then for any
sensor $s_j$ in $S_R$ before $s_{g(i)}$, the left extension of $s_j$
must be larger than $p_1+2\lambda$; we call the sensors in $S_R$
before $s_{g(i)}$ the {\em redundant sensors} for the step $i$
(see Fig.~\ref{fig:redundant}). Later
we will show that these sensors will be not redundant any more in the
later algorithm.
In summary, for each sensor scanned in the original $S_R$ in this
step, it is removed, or a redundant sensor, or $s_{g(i)}$. Finally, we
move $s_{g(i)}$ to the left such that its left extension is at $p_1$,
and then we move $p_1$ to the right extension of $s_{g(i)}$ (i.e.,
$p_1=R_i$). If $p_1\geq L$, we terminate the algorithm and report
$\lambda^*\leq \lambda$; otherwise, we continue on the next step
$i+1$.

To analyze the algorithm, it is easy to see each sensor in this list
$S_L$ is scanned at most once. For the list $S_R$, this may not be the
case as the redundant sensors may be scanned again in the later
algorithm. However, the following lemma shows that this would not be
an issue.

\begin{lemma}
If a sensor $s_j$ is a redundant sensor for the step $i$, then it will be not a redundant sensor again in the later algorithm.
\end{lemma}
\begin{proof}
Consider the moment right after the step $i$. The sweeping point $p_1$ is at the right extension of $s_{g(i)}$.
To prove the lemma, since $p_1$ always moves to the right, by the definition of the redundant sensors, it is sufficient to show that the left extension of $s_j$ is at most $p_1+2\lambda$, as follows.

Consider the moment in the beginning of the step $i$ (the sensor $s_{g(i)}$ has not been moved to the left). Since $s_j$ is a redundant sensor for the step $i$, the sensor $s_{g(i)}$ is from $S_{i2}$ and the left extension of $s_{g(i)}$ is at most $R_{i-1}+2\lambda$. Thus, the right extension of $s_{g(i)}$ is at most $R_{i-1}+2r_{g(i)}+2\lambda$. Recall that the right extension of $s_j$ is less than that of $s_{g(i)}$ (since $s_j$ is before $s_{g(i)}$ in $S_R$). Therefore, the right extension of $s_j$ is at most $R_{i-1}+2r_{g(i)}+2\lambda$. Now consider the moment right after the step $i$. The sweeping point $p_1$ is at the position $R_{i-1}+2r_{g(i)}$. Hence, the right extension of $s_j$ is at most $p_1+2\lambda$, which implies that the left extension of $s_j$ is at most $p_1+2\lambda$. The lemma thus follows.
\end{proof}

The preceding lemma implies that any sensor can be a redundant sensor in at most one step. Therefore, for the list $S_R$, each sensor has been scanned at most twice when it is removed, once as a redundant sensor, and once when it is found as $s_{g(i)}$. Thus, each sensor in $S_R$ is scanned at most three times. Hence, after the two lists $S_L$ and $S_R$ are obtained, the running time of the algorithm is $O(n)$.

\begin{theorem}\label{theo:generaldecision}
After $O(n\log n)$ time preprocessing, for any $\lambda$, we can determine whether
$\lambda^*\leq \lambda$ in $O(n)$ time; further, if
$\lambda^*\leq \lambda$, we can compute a feasible solution in $O(n)$ time.
\end{theorem}

\subsubsection{Another Decision Version}
\label{subsec:more}

Our optimization algorithm in Section \ref{sec:algogeneral} also needs to
determine whether $\lambda^*$ is strictly less than $\lambda$ (i.e.,
$\lambda^*<\lambda$) for any $\lambda$. By
modifying our algorithm for Theorem \ref{theo:generaldecision}, we have
the following result.

\begin{theorem}\label{theo:generaldecision2}
After $O(n\log n)$ time preprocessing, for any value $\lambda$,
we can determine whether $\lambda^*<\lambda$ in $O(n)$ time.
\end{theorem}
\begin{proof}
We first apply the algorithm for Theorem \ref{theo:generaldecision} on
the value $\lambda$. If the algorithm reports $\lambda^*>\lambda$, then
we know $\lambda^*<\lambda$ is false.
Otherwise, we have $\lambda^*\leq \lambda$. In the following, we modify
the algorithm for Theorem \ref{theo:generaldecision}
to determine whether $\lambda^*<\lambda$, i.e., $\lambda^*$ is strictly
smaller than $\lambda$. Note that this is equivalent to deciding whether
$\lambda^*\leq \lambda-\epsilon$ for any arbitrarily small constant
$\epsilon>0$. Of course, we cannot enumerate all such small values
$\epsilon$. Instead, we add a new mechanism to the algorithm for
Theorem \ref{theo:generaldecision} such that the resulting
displacement of each sensor is strictly smaller than $\lambda$.

At the start of the algorithm, we move all sensors to the right by a
distance $\lambda$ to obtain the configuration $C_0$. But, the displacement
of each sensor should be strictly less than  $\lambda$. To ensure this,
later in the algorithm, if
the destination of a sensor $s_i$ is set as $y_i=x_i'$, then we
adjust this destination of $s_i$ by moving it to the left slightly
such that $s_i$'s displacement is strictly less than $\lambda$.

Consider a general step $i$ of the algorithm. We define the set
$S_{i1}$ in the same way as before, i.e., it consists of all sensors
covering the point $p^+(R_{i-1})$ in $C_{i-1}$. If $S_{i1}\neq\emptyset$,
then the algorithm is the same as before. In this case, the sensor
$s_{g(i)}$ chosen in this step has a displacement of exactly $\lambda$,
which is actually ``illegal" since the displacement of each sensor should be
strictly less than $\lambda$. We will address this issue later.
However, if $S_{i1}=\emptyset$, then
the set $S_{i2}$ is defined slightly different from before. Here, since
$S_{i1}=\emptyset$, we have to use a sensor to the right of
$R_{i-1}$ in $C_{i-1}$ to cover $p^+(R_{i-1})$. Since the displacement
of each sensor should be strictly less than $\lambda$, we do not allow any
sensor to move to the left by exactly the distance $2\lambda$.
To reflect this difference, we define $S_{i2}$ as the set of sensors in
$C_{i-1}$ each of which has its left extension larger than $R_{i-1}$ and
{\em strictly smaller} than $R_{i-1}+2\lambda$ (previously, it was ``at most").
In this way, if we move a sensor in $S_{i2}$ to the left to cover
$p^+(R_{i-1})$, then the displacement of that sensor is strictly less than
$\lambda$.  The rest of the algorithm is the same as before.
We define the type I and type II sensors in the same way as before.

If the algorithm terminates without finding a feasible solution, then it must
be $\lambda^*\geq \lambda$; otherwise, the algorithm finds a ``feasible"
solution $SOL$ with a critical set
$S^c=\{s_{g(1)},s_{g(2)},\ldots,s_{g(m)}\}$. But, this does not
necessarily mean $\lambda^*<\lambda$ since in $SOL$, the displacements
of some sensors in $S^c$ may be exactly $\lambda$. Specifically, all
type I sensors in $S^c$ are in the same positions as they are in $C_0$
and thus their displacements are exactly $\lambda$. In contrast,
during the algorithm, the type II sensors in $S^c$ have been moved
strictly to the left with respect to their positions in $C_0$; further,
due to our new definition of the set $S_{i2}$, the displacements
of all type II sensors are strictly less than $\lambda$. Therefore, if
there is no type I sensor in $S^c$, then the displacement of each
sensor in $S^c$ is strictly less than $\lambda$ and thus we have
$\lambda^*<\lambda$. Below we assume $S^c$ contains at least
one type I sensor.  To make sure that $\lambda^*<\lambda$ holds,
we need to find a {\em real feasible solution} in which the displacement
of each sensor in $S$ is strictly less than $\lambda$. On the other hand,
to make sure that $\lambda^*\geq \lambda$ holds, we must show that there is no
real feasible solution.
For this, we apply the following algorithmic procedure.

We seek to adjust the solution $SOL$ to produce a real feasible
solution. According to our algorithm, for each sensor $s_i\in S^c$, if
it is a type I sensor, then $y_i=x_i'$ and thus its displacement is
exactly $\lambda$; otherwise, its displacement is less than
$\lambda$. The purpose of our adjustment of $SOL$ is to move all
type I sensors slightly to the left so that (1) their displacements are
strictly less than $\lambda$, and (2) we can still form a coverage of $B$. In
certain cases, we may need to use some sensors in $S\setminus S^c$ as well.
Also, we may end up with the conclusion that no real feasible solution
exists.

According to our algorithm, after finding the last sensor $s_{g(m)}$ in $S^c$,
we have $R_m\geq L$. If $R_m>L$, then we can always adjust $SOL$ to obtain a real
feasible solution by shifting each sensor in $S^c$ to the left by a very
small value $\epsilon$ such that (1) the resulting displacement of each sensor
in $S^c$ is less than $\lambda$, and (2) the sensors of $S^c$ still form a coverage of
$B$. Note that there always exists such a small value $\epsilon$ such that the above adjustment is  possible. Therefore, if $R_m>L$, then
we have $\lambda^*<\lambda$.

If $R_m=L$, however, then the above strategy does not work. There are two cases.
If there is a sensor $s_t\in S\setminus S^c$ such that $x_t\in
(L-\lambda-r_t,L+\lambda+r_t)$, then we can also obtain a real feasible solution
%
%
%
by shifting the sensors of $S^c$ slightly to the left as above and using the
sensor $s_t$ to cover the remaining part of $B$ around $L$ that is no
longer covered
by the shifted sensors of $S^c$; thus we also have $\lambda^*<\lambda$.
Otherwise, we claim that it must be $\lambda^*\geq \lambda$.
Below we prove this claim.

Consider the rightmost Type I sensor $s_i$ in $S^c$. Suppose
$s_i=s_{g(j)}$, i.e., $s_i$ is determined in step $j$. Thus,
$s_i$ is at $x_i'$ in $SOL$. Let $\epsilon>0$ be an arbitrarily small
value (we will determine below how small it should be).  Since we
have assumed that the extensions of all sensors are different,
the value $\epsilon$ can be made small enough such that by moving $s_i$ to
$x_i'-\epsilon$ in $C_0$, the relative order of
the extensions of all sensors remains the same as before.
Further, according to our algorithm above, the value $\epsilon$ can also be
small enough such that the behavior of the algorithm is the same as
before, i.e., the algorithm finds the same critical set $S^c$ with the same
cover order as before.
It is easy to see that such a small value $\epsilon$ always exists.
Note that our task here is to prove our claim $\lambda^*\geq
\lambda$ is true, and knowing that such a value
$\epsilon$ exists is sufficient for our purpose
and we need not actually find such a value $\epsilon$ in our algorithm.

Now, in step $j$, the new value $R_j$, which is the right extension of
$s_i$, is $\epsilon$ smaller than its value before since $s_i$ was at
$x'_i$ in $C_0$. Because $s_i$ is the rightmost type I sensor in $S^c$,
after step $j$, all sensors in $S^c$ determined after $s_i$ (if any) are
of type II and thus are moved to the left such that
they are all in attached positions along with $s_i$, which implies that the
right extension of the last sensor $s_{g(m)}$ in $S^c$ is also $\epsilon$
smaller than its previous value (which was $L$). Hence, after
step $m$, the sensors in $S_m$ covers $[0,L-\epsilon]$.
As discussed above, if $\epsilon$ is made small enough, the behavior of the
algorithm is the same as before. By a similar analysis,
we can also establish a result similar
to Lemma \ref{lem:correct20}. Namely, $[0,L-\epsilon]$ is the largest
left-aligned interval that can be covered by the sensors in $S_m'$ in this
setting (here, $S_m'$ is the set of sensors whose right extensions are at most
$L-\epsilon$ in the configuration after step $m$).
We omit the detailed analysis for this, which is very similar to that for
Lemma \ref{lem:correct20}.
Note that $S^c=S_m$.
Since there is no sensor $s_t\in S\setminus S^c$ such that $x_t\in
%
%
%
(L-\lambda-r_t,L+\lambda+r_t)$, the interval $(L-\epsilon,L]$ cannot be fully covered
by the sensors in $S$.
%
%
%
The above discussion implies that if we do not allow the displacement
of $s_i$ to be larger than $\lambda-\epsilon$, then there would be no
feasible solution even if we allow the displacements of some other
sensors (i.e., those type I sensors in $S^c$ before $s_i$, if any) to be
larger than $\lambda-\epsilon$ (but at most $\lambda$). Thus,
$\lambda^*\leq \lambda-\epsilon$ cannot be true. That is,
$\lambda^*>\lambda-\epsilon$ holds. Further, it is easy to see that,
by a similar argument, for any fixed value $\epsilon'>0$ with
$\epsilon'<\epsilon$, we also have $\lambda^*>\lambda-\epsilon'$.
Hence, we obtain $\lambda^*\geq\lambda$.

This finishes the discussion on how to determine whether
$\lambda^*<\lambda$. It is easy to see that the above algorithm can also be
implemented in $O(n)$ time for each value $\lambda$,
after $O(n\log n)$ time preprocessing. The theorem thus follows.
\end{proof}

Theorems \ref{theo:generaldecision} and \ref{theo:generaldecision2}
together lead to the following corollary.

\begin{corollary}\label{cor:20}
After $O(n\log n)$ time preprocessing, for any value $\lambda$, we can
determine whether $\lambda^*=\lambda$ in $O(n)$ time.
\end{corollary}

\subsection{The Optimization Version of the General \bcls}
\label{sec:algogeneral}

In this section, we discuss the optimization version of the general
\bcls\ problem. We show that it is solvable in $O(n^2\log
n)$ time, thus settling the open problem in \cite{ref:CzyzowiczOn09}.

It should be pointed out that if we can determine a set $\Lambda$ of
candidate values such that $\lambda^*\in \Lambda$, then we would use
our decision algorithms given in Section \ref{sec:decision} to find
$\lambda^*$ in $\Lambda$. We will use this approach in Section
\ref{sec:uniform} for the uniform case. However, so far it is not clear to
us how to determine such a set $\Lambda$. Below, we use a different
approach.

One main difficulty for solving the problem is that we do not know the order of the
sensors in the optimal solution.
Our strategy is to determine a critical set of sensors and their
cover order in a feasible solution for
the (unknown) optimal value $\lambda^*$.
The idea is somewhat similar to parametric
search \cite{ref:ColeSl87,ref:MegiddoAp83} and here we ``parameterize"
our algorithm for Theorem \ref{theo:generaldecision}. But, unlike the typical
parametric search \cite{ref:ColeSl87,ref:MegiddoAp83},
our approach does not involve any parallel scheme and is practical.
We first give an overview of this algorithm. In the following
discussion, the ``decision algorithm" refers to our algorithm
for Theorem \ref{theo:generaldecision} unless otherwise stated.

Recall that given any value $\lambda$, step $i$ of our decision
algorithm determines the sensor $s_{g(i)}$ and obtains the set
$S_i=\{s_{g(1)},s_{g(2)},\ldots,s_{g(i)}\}$, in this order, which
we also call the {\em cover order} of the sensors
in $S_i$. In our optimization algorithm, we often use
$\lambda$ as a variable. Thus, $S_i(\lambda)$ (resp.,
$R_i(\lambda)$, $s_{g(i)}(\lambda)$, and $C_i(\lambda)$) refers to the
corresponding $S_i$ (resp., $R_i$, $s_{g(i)}$, and $C_i$) obtained by
running our decision algorithm on the specific value $\lambda$. Denote by $C_I$
the configuration of the input.

Our optimization algorithm takes at most $n$ steps.
Initially, let $S_0(\lambda^*)=\emptyset$, $R_0(\lambda^*)=0$,
$\lambda^1_0=0$, and $\lambda^2_0=+\infty$.
For each $i\geq 1$, step $i$ receives an
interval $(\lambda^1_{i-1},\lambda^2_{i-1})$ and a sensor set $S_{i-1}
(\lambda^*)$, with the following {\em algorithm invariants}:

\begin{itemize}
\item
$\lambda^*\in (\lambda^1_{i-1},\lambda^2_{i-1})$.

\item
For any value $\lambda\in (\lambda^1_{i-1},\lambda^2_{i-1})$,
$S_{i-1}(\lambda)=S_{i-1} (\lambda^*)$ and their cover orders are
the same.
\end{itemize}

Step $i$ either finds the value
$\lambda^*$ or determines a sensor
$s_{g(i)}(\lambda^*)$.  The interval
$(\lambda^1_{i-1},\lambda^2_{i-1})$ will shrink to a new
interval $(\lambda^1_{i},\lambda^2_{i})\subseteq
(\lambda^1_{i-1},\lambda^2_{i-1})$ and we also obtain the set
$S_i(\lambda^*)=S_{i-1}(\lambda^*)\cup \{s_{g(i)}(\lambda^*)\}$.
All these can be done in $O(n\log n)$ time.
The details of the algorithm are given below.

Consider a general step $i$ for $i\geq 1$ and we have the interval
$(\lambda^1_{i-1},\lambda^2_{i-1})$ and the set $S_{i-1}(\lambda^*)$.
While discussing the algorithm, we will also prove inductively the
following lemma about the function $R_{i}(\lambda)$ with variable
$\lambda\in (\lambda_{i}^1,\lambda_{i}^2)$.

\begin{lemma}\label{lem:new}
For any step $i$ with $i\geq 0$, if the algorithm does not stop after the step, then the following hold:
\begin{description}
\item[(a)] The function $R_{i}(\lambda)$ for $\lambda\in (\lambda_{i}^1,\lambda_{i}^2)$
is a line segment of slope $1$ or
$0$.
\item[(b)] We can compute the function $R_{i}(\lambda)$ for $\lambda\in(\lambda_{i}^1,\lambda_{i}^2)$
 explicitly in $O(n)$ time.
\item[(c)] $R_{i}(\lambda)<L$ for any $\lambda\in(\lambda_{i}^1,\lambda_{i}^2)$.
\end{description}
\end{lemma}

In the base case for $i=0$, the statement of Lemma \ref{lem:new}
obviously holds. We assume the lemma statement holds for $i-1$, in particular, the function $R_{i-1}(\lambda)$ for $\lambda\in(\lambda_{i-1}^1,\lambda_{i-1}^2)$ is already known. We
will show that after step $i$ the lemma statement holds for $i$, and
thus the lemma will be proved.

Again, in step $i$, we need to determine the sensor $s_{g(i)}(\lambda^*)$ and let
$S_i(\lambda^*)=S_{i-1}(\lambda^*)\cup \{s_{g(i)}(\lambda^*)\}$. We
will also obtain an interval $(\lambda_i^1,\lambda_i^2)$ such that
$\lambda^*\in(\lambda_i^1,\lambda_i^2)\subseteq
(\lambda_{i-1}^1,\lambda_{i-1}^2)$ and for any $\lambda\in
(\lambda_i^1,\lambda_i^2)$, $S_i(\lambda) = S_i(\lambda^*)$ holds
(with the same cover order).

\begin{figure}[t]
\begin{minipage}[t]{\linewidth}
\begin{center}
\includegraphics[totalheight=1.1in]{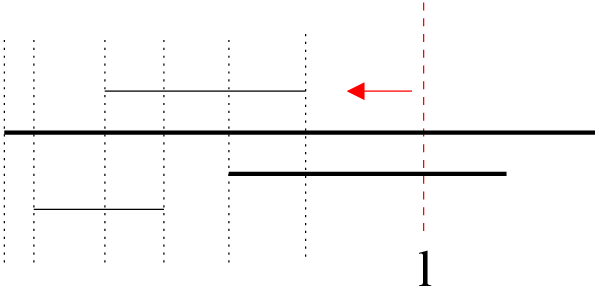}
\caption{\footnotesize The vertical dashed line $l$ is
$x=R_{i-1}(\lambda)$. The set $S_{i1}(\lambda)$ consists of the
sensors whose covering intervals intersect $l$, shown in bold. As
$\lambda$ increases, the relative position of $l$ in
$C_{i-1}(\lambda)$ moves to the left, and $S_{i1}(\lambda)$ changes
whenever $l$ hits a sensor extension, shown with the dotted vertical lines.}
\label{fig:relativepos}
\end{center}
\end{minipage}
\vspace*{-0.15in}
\end{figure}

To find the sensor $s_{g(i)}(\lambda^*)$, we first determine the
set $S_{i1}(\lambda^*)$. Recall that $S_{i1}(\lambda^*)$ consists of
all sensors covering the point $p^+(R_{i-1}(\lambda^*))$ in the
configuration $C_{i-1}(\lambda^*)$.
For each sensor in $S\setminus S_{i-1}(\lambda^*)$, its position in the
configuration $C_{i-1}(\lambda)$ with respect to
$\lambda\in(\lambda_{i-1}^1,\lambda_{i-1}^2)$ is a function of slope
$1$. As $\lambda$ increases in $(\lambda_{i-1}^1,\lambda_{i-1}^2)$, by our assumption that Lemma \ref{lem:new}(a) holds for $i-1$, the function $R_{i-1}(\lambda)$ is a line segment of slope $1$ or $0$. If
$R_{i-1}(\lambda)$ is of slope $1$, then the relative
position of $R_{i-1}(\lambda)$ in $C_{i-1}(\lambda)$ does not change and
thus the set $S_{i1}(\lambda)$ does not change; if the function
$R_{i-1}(\lambda)$ is of slope $0$, then the relative
position of $R_{i-1}(\lambda)$ in $C_{i-1}(\lambda)$ is monotonically moving to the
left. Hence, there are $O(n)$ values for $\lambda$ in
$(\lambda_{i-1}^1,\lambda_{i-1}^2)$ that can incur some changes to the
set $S_{i1}(\lambda)$ and each such value corresponds to a sensor
extension (e.g., see Fig.~\ref{fig:relativepos});
further, these values can be easily determined in $O(n\log n)$
time by a simple sweeping process (we omit the discussion of it). Let $\Lambda_{i1}$ be the set of all these
$\lambda$ values. Let $\Lambda_{i1}$ also contain both $\lambda_{i-1}^1$ and $\lambda_{i-1}^2$, and thus,
$\lambda_{i-1}^1$ and $\lambda_{i-1}^2$ are the smallest and
largest values in $\Lambda_{i1}$, respectively. We sort the values in $\Lambda_{i1}$. For any
two consecutive values $\lambda_1<\lambda_2$ in the sorted $\Lambda_{i1}$,
the set $S_{i1}(\lambda)$ for any
$\lambda\in(\lambda_1,\lambda_2)$ is the same. By using binary search on the sorted $\Lambda_{i1}$ and our
decision algorithm in Theorem \ref{theo:generaldecision}, we determine (in $O(n\log n)$ time) the two
consecutive values $\lambda_1$ and $\lambda_2$ in $\Lambda_{i1}$ such that
$\lambda_1<\lambda^*\leq \lambda_2$. Further, by Corollary
\ref{cor:20}, we determine whether $\lambda^*=\lambda_2$. If
$\lambda^*=\lambda_2$, then we terminate the algorithm.
Otherwise, based on our discussion above, $S_{i1}(\lambda^*)= S_{i1}(\lambda)$ for any
$\lambda\in(\lambda_1,\lambda_2)$. Thus, to
compute $S_{i1}(\lambda^*)$, we can pick an arbitrary $\lambda$ in
$(\lambda_1,\lambda_2)$ and find $S_{i1}(\lambda)$ in the same way as
in our decision algorithm. Hence, $S_{i1}(\lambda^*)$ can be easily found in $O(n\log n)$ time. Note that
$\lambda^*\in(\lambda_1,\lambda_2)\subseteq(\lambda_{i-1}^1,\lambda_{i-1}^2)$. Depending on whether $S_{i1}(\lambda^*)\neq\emptyset$, there are two cases.

\begin{itemize}
\item
If $S_{i1}(\lambda^*)\neq\emptyset$, then $s_{g(i)}(\lambda^*)$ is the
sensor in $S_{i1}(\lambda^*)$ with the largest right extension. An obvious observation is that for any $\lambda\in (\lambda_1,\lambda_2)$, the sensor in
$S_{i1}(\lambda)$ with the largest right extension is the same,
which can be easily found.
We let $\lambda^1_i=\lambda_1$ and $\lambda^2_i=\lambda_2$. Let
$S_i(\lambda^*)=S_{i-1}(\lambda^*)\cup\{s_{g(i)}(\lambda^*)\}$.
The algorithm invariants hold.
Further, as $\lambda$ increases in $(\lambda^1_i,\lambda^2_i)$, the
right extension of $s_{g(i)}(\lambda)$, which is $R_i(\lambda)$,
increases by the same amount.  That is, the function
$R_i(\lambda)$ on $(\lambda^1_i,\lambda^2_i)$ is a line segment of slope $1$.
Therefore, we can compute $R_i(\lambda)$ on $(\lambda^1_i,\lambda^2_i)$
explicitly in constant time. This also shows
Lemma \ref{lem:new}(a) and (b) hold for $i$.

\item
If $S_{i1}(\lambda^*)=\emptyset$, then we need to compute $S_{i2}(\lambda^*)$.
For any $\lambda\in (\lambda_1,\lambda_2)$, the set
$S_{i2}(\lambda)$ consists of all sensors whose left extensions are
larger than $R_{i-1}(\lambda)$ and at most
$R_{i-1}(\lambda)+2\lambda$ in the configuration $C_{i-1}(\lambda)$.
Recall that the function $R_{i-1}(\lambda)$ on
$(\lambda_{i-1}^1,\lambda_{i-1}^2)$ is linear with slope $1$ or $0$.
Due to $(\lambda_1,\lambda_2)\subseteq (\lambda_{i-1}^1,\lambda_{i-1}^2)$,
the linear function $R_{i-1}(\lambda)+2\lambda$ on $(\lambda_1,\lambda_2)$
is of slope $3$ or $2$.  Again, as $\lambda$ increases, the
position of each sensor in $S\setminus S_{i-1}(\lambda^*)$
in $C_{i-1}(\lambda)$ is a linear function of slope
$1$. Hence, as $\lambda$ increases, the relative position of $R_{i-1}(\lambda)+2\lambda$ in $C_{i-1}(\lambda)$ moves to the right, and the relative position of $R_{i-1}(\lambda)$ in $C_{i-1}(\lambda)$ either does not change or moves to the left.
Therefore, there are $O(n)$ $\lambda$ values in
$(\lambda_1,\lambda_2)$ each of which incurs
some change to the set $S_{i2}(\lambda)$ and each such
$\lambda$ value corresponds to the left extension of a sensor (e.g., see Fig.~\ref{fig:relativepos2}).
Further,
these values can be easily determined in $O(n\log n)$ time by a sweeping process (we omit the discussion for this). (Actually,
as $\lambda$ increases, the size of the set $S_{i2}(\lambda)$ is monotonically
increasing.) Let $\Lambda_{i2}$ denote the set of these $\lambda$ values,
and let $\Lambda_{i2}$ contain both $\lambda_1$ and $\lambda_2$.
Again, $|\Lambda_{i2}|=O(n)$. We sort the values in $\Lambda_{i2}$.
Using binary search on the sorted $\Lambda_{i2}$ and our
decision algorithm in Theorem \ref{theo:generaldecision}, we determine (in $O(n\log n)$ time) the two
consecutive values $\lambda'_1$ and $\lambda'_2$ in $\Lambda_{i2}$ such that
$\lambda'_1<\lambda^*\leq \lambda'_2$. Further, by Corollary
\ref{cor:20}, we determine whether $\lambda^*=\lambda'_2$. If
$\lambda^*=\lambda'_2$, then we are done.
Otherwise, $S_{i2}(\lambda^*)=
S_{i2}(\lambda)$ for any $\lambda\in (\lambda'_1,\lambda'_2)$,
which can be easily found.
Note that $\lambda^*\in (\lambda'_1,\lambda'_2)\subseteq (\lambda_1,\lambda_2)$.

\begin{figure}[t]
\begin{minipage}[t]{\linewidth}
\begin{center}
\includegraphics[totalheight=1.1in]{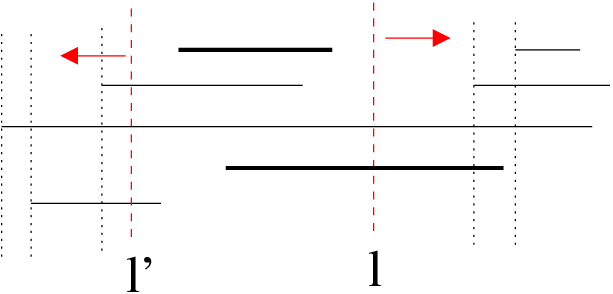}
\caption{\footnotesize The two vertical dashed lines $l$ is $x=R_{i-1}(\lambda)+2\lambda$ and $l'$ is $x=R_{i-1}(\lambda)$. The set $S_{i2}(\lambda)$ consists of sensors whose left extensions are between $l$ and $l'$, shown in bold. As $\lambda$ increases, the relative position of $l$ in $C_{i-1}(\lambda)$ moves to the right, and the relative position of $l'$ in $C_{i-1}(\lambda)$ either does not change or moves to the left. The set $S_{i2}(\lambda)$ changes whenever $l$ or $l'$ hits the left extension of a sensor, shown with the dotted vertical lines.}
\label{fig:relativepos2}
\end{center}
\end{minipage}
\vspace*{-0.15in}
\end{figure}

The above obtains the set $S_{i2}(\lambda^*)$. We claim that
$S_{i2}(\lambda^*)\neq \emptyset$. Indeed, due to our assumption that
Lemma 4 holds for $i-1$, we have $R_{i-1}(\lambda)<L$ for
$\lambda\in(\lambda_{i-1}^1,\lambda_{i-1}^2)$. Suppose to the contrary
that $S_{i2}(\lambda^*)= \emptyset$. Then, the sensor
$s_{g(i)}(\lambda^*)$ does not exist, which implies that
$S_{i-1}(\lambda^*)$ is the critical set for covering the barrier $B$
in an optimal solution. By our algorithm invariants, $\lambda^*\in
(\lambda_{i-1}^1,\lambda_{i-1}^2)$ and $S_{i-1}(\lambda^*)$ is the
same as $S_{i-1}(\lambda)$ for any $\lambda\in
(\lambda_{i-1}^1,\lambda_{i-1}^2)$. Due to $R_{i-1}(\lambda)<L$ for
$\lambda\in(\lambda_{i-1}^1,\lambda_{i-1}^2)$, the sensors in
$S_{i-1}(\lambda^*)$ cannot cover the entire barrier $B$, which
contradicts with that $S_{i-1}(\lambda^*)$ is the critical set in the
optimal solution. Hence, $S_{i2}(\lambda^*)\neq \emptyset$.

Since $S_{i2}(\lambda^*)\neq \emptyset$, $s_{g(i)}(\lambda^*)$
is the sensor in $S_{i2}(\lambda^*)$ with
the smallest right extension. As before, the sensor in $S_{i2}(\lambda)$ with
the smallest right extension is the same for any $\lambda\in
(\lambda'_1,\lambda'_2)$. Thus, $s_{g(i)}(\lambda^*)$ can be easily
determined. We let $\lambda^1_i=\lambda'_1$ and $\lambda^2_i=\lambda'_2$.
Let $S_i(\lambda^*)=S_{i-1}(\lambda^*)\cup\{s_{g(i)}(\lambda^*)\}$.
The algorithm invariants hold.
Further, we examine the function $R_i(\lambda)$, i.e., the
right extension of $s_{g(i)}(\lambda)$ in the configuration $C_i(\lambda)$, as
$\lambda$ increases in $(\lambda^1_i,\lambda^2_i)$. Since
$s_{g(i-1)}(\lambda^*)$ and $s_{g(i)}(\lambda^*)$ are always in attached
positions in this case, for any $\lambda\in (\lambda^1_i,\lambda^2_i)$, we have
$R_i(\lambda)=R_{i-1}(\lambda)+2r_{g(i)}$. Thus, the function $R_i(\lambda)$ is a vertical shift of $R_{i-1}(\lambda)$ by the distance $2r_{g(i)}$.
 Because we already know
explicitly the function $R_{i-1}(\lambda)$ for $\lambda\in
(\lambda^1_i,\lambda^2_i)$, which is a line segment of
slope $1$ or $0$, the function $R_i(\lambda)$ can be computed in constant time, which is also a line segment of slope $1$ or $0$. Note that this shows that Lemma \ref{lem:new}(a) and (b) hold for $i$.
\end{itemize}

If the algorithm does not stop, the above determines an interval $(\lambda^1_i,\lambda^2_i)$ such that the algorithm invariants and Lemma 4(a) and (b) hold on the interval. Below, we do further processing such that Lemma 4(c) also holds.

Because the function
$R_i(\lambda)$ on $(\lambda^1_i,\lambda^2_i)$ is a line segment of slope $1$ or $0$,
there are three cases depending on the values $R_i(\lambda)$ and $L$: (1)
$R_i(\lambda)<L$ for any $\lambda\in (\lambda^1_i,\lambda^2_i)$, (2)
$R_i(\lambda)>L$ for any $\lambda\in (\lambda^1_i,\lambda^2_i)$, and (3)
there exists $\lambda'\in
(\lambda^1_i,\lambda^2_i)$ such that $R_i(\lambda')=L$.

\begin{enumerate}
\item
For Case (1), we proceed to the next step,
along with the interval $(\lambda^1_i,\lambda^2_i)$.  Clearly, the
algorithm invariants hold and Lemma \ref{lem:new}(c) holds for $i$.

\item
For Case (2), the next lemma shows that it actually cannot happen
due to $\lambda^*\in(\lambda^1_i,\lambda^2_i)$.

\begin{lemma}\label{lem:lattercase}
It is not possible that $R_i(\lambda)>L$ for any $\lambda\in
(\lambda^1_i,\lambda^2_i)$.
\end{lemma}
\begin{proof}
Assume to the contrary that $R_i(\lambda)>L$ for any $\lambda\in
(\lambda^1_i,\lambda^2_i)$. Since $\lambda^*\in(\lambda^1_i,\lambda^2_i)$, let
$\lambda''$ be any value in $(\lambda^1_i,\lambda^*)$. Due to
$\lambda''\in(\lambda^1_i,\lambda^2_i)$, we have $R_i(\lambda'')>L$. But this would
implies that we have found a feasible solution where the displacement
of each sensor is at most $\lambda''$, which is smaller than $\lambda^*$, incurring
contradiction.
\end{proof}

\item
For the Case (3), note that the slope of $R_i(\lambda)$ on $(\lambda^1_i,\lambda^2_i)$ can not be $0$. To see this, suppose to the contrary the slope of $R_i(\lambda)$ on $(\lambda^1_i,\lambda^2_i)$ is $0$. Then, $R_i(\lambda)=L$ for any $\lambda\in (\lambda^1_i,\lambda^2_i)$. Since $\lambda^*\in
(\lambda^1_i,\lambda^2_i)$, for any $\lambda'\in (\lambda^1_i,\lambda^*)$, $R_i(\lambda')=L$, which means that there is a feasible solution where the displacement of each sensor is at most $\lambda'<\lambda^*$, incurring contradiction.

Hence, $R_i(\lambda)$ on $(\lambda^1_i,\lambda^2_i)$ is a line
segment of slope $1$, and thus
we can determine in constant time the unique value $\lambda'\in
(\lambda^1_i,\lambda^2_i)$ such that $R_i(\lambda')=L$.
Clearly, $\lambda^*\leq \lambda'$. By Corollary
\ref{cor:20}, we determine whether $\lambda^*=\lambda'$. If
$\lambda^*=\lambda'$, then we terminate the algorithm; otherwise, we have $\lambda^*\in
(\lambda^1_i,\lambda')$ and update $\lambda^2_i$ to $\lambda'$.  We proceed to the
next step, along with the interval $(\lambda^1_i,\lambda^2_i)$.
Again, the algorithm invariants hold and Lemma \ref{lem:new}(c) holds for $i$.
\end{enumerate}


This finishes the discussion of step $i$ of our algorithm. The running time of step $i$ is $O(n\log n)$. Note that in each case where we proceed to the next step, the statement of Lemma \ref{lem:new} holds for $i$, and thus Lemma \ref{lem:new} has been proved.

In the following lemma, we show that the algorithm must stop
within at most $n$ steps.

\begin{lemma}
The algorithm finds $\lambda^*$ in at most $n$ steps.
\end{lemma}
\begin{proof}
Assume the critical set is $S_k(\lambda^*)$ for some $k$ if we run our decision algorithm with $\lambda=\lambda^*$. Since there are $n$ sensors from the input, we have $1\leq k\leq n$.

We claim that our algorithm finds $\lambda^*$ in at most $k$ steps. Suppose to the contrary that the algorithm does not find $\lambda^*$ in the first $k$ steps. In other words, the algorithm does not stop after step $k$. By the algorithm invariants, after step $k$, we have an interval $(\lambda_k^1,\lambda_k^2)$ such that $\lambda^* \in (\lambda_k^1,\lambda_k^2)$ and $S_k(\lambda)=S_k(\lambda^*)$ for any $\lambda \in (\lambda_k^1,\lambda_k^2)$. Further, by Lemma \ref{lem:new}, $R_k(\lambda)<L$ for any $\lambda \in (\lambda_k^1,\lambda_k^2)$, which means that the sensors in $S_k(\lambda^*)$ cannot cover the entire barrier $B$ for any $\lambda \in (\lambda_k^1,\lambda_k^2)$, contradicting with that $S_k(\lambda^*)$ is the critical set for the decision algorithm when $\lambda=\lambda^*$.

Therefore, our algorithm finds $\lambda^*$ in at most $k$ steps. The lemma thus follows. \end{proof}

After $\lambda^*$ is found,
by applying our decision algorithm on $\lambda=\lambda^*$, we
finally produce an optimal solution in which
the displacement of every sensor is at most $\lambda^*$.  Since each
step takes $O(n\log n)$ time, the total time of the algorithm is
$O(n^2\log n)$.

\begin{theorem}
The general \bcls\ problem is solvable in $O(n^2\log n)$ time.
\end{theorem}

We shall make a technical remark. The typical parametric search
\cite{ref:ColeSl87,ref:MegiddoAp83} usually returns with an interval
containing the optimal value and then uses an additional step to find
the optimal value. In contrast, our
algorithm is guaranteed to find the optimal value $\lambda^*$
directly. This is due to
the mechanism in our algorithm that requires $R_i(\lambda)<L$ for
any $\lambda\in (\lambda^1_i,\lambda^2_i)$
after each step $i$ if the algorithm is not terminated.
This mechanism actually plays the role of the
additional step used in the typical parametric search.

\section{The Uniform Case of \bcls}
\label{sec:uniform}

In this section, we present an $O(n\log n)$ time algorithm for the uniform
case of \bcls. Previously, the best known algorithm for it takes $O(n^2)$
time \cite{ref:CzyzowiczOn09}. Further, for the special uniform case when all
sensors are initially located on the barrier $B$, we solve it in $O(n)$ time.

\subsection{Preliminaries}

Recall that in the input, all sensors are ordered from left to right by their
initial positions, i.e., $x_1\leq x_2\leq\cdots\leq x_n$.
Suppose in a solution, the destination of each sensor
$s_i$ is $y_i$ ($1\leq i\leq n$); then we say that the solution is {\em
order preserving} if $y_1\leq y_2\leq \cdots\leq y_n$.
In the uniform case, since all sensors have the same range,
we let $r$ denote the sensor range.
The next lemma was known \cite{ref:CzyzowiczOn09}.

\begin{lemma}\label{lem:10}
{\em (Czyzowicz {\em et al.} \cite{ref:CzyzowiczOn09})}
For the uniform case, there is always an optimal solution that is order preserving.
\end{lemma}

As discussed in \cite{ref:CzyzowiczOn09},
Lemma \ref{lem:10} is not applicable to the general \bcls.
Consequently, the approach in this
section does not work for the general \bcls.

Based on the order preserving property in Lemma \ref{lem:10},
the previous $O(n^2)$ time algorithm
\cite{ref:CzyzowiczOn09} tries to cover $B$ from left to right;
each step picks the next sensor and re-balances the current maximum
sensor movement. Here, we take a completely different approach.

Denote by $\lambda^*$ the maximum sensor movement in an optimal
solution. We use $OPT$ to denote an optimal order
preserving solution in which the destination for each sensor
$s_i$ is $y_i$ ($1\leq i\leq n$). For each sensor $s_i$, if $x_i>y_i$
(resp., $x_i<y_i$), then we say $s_i$ is moved to the left (resp., right)
by a distance $|x_i-y_i|$.
A set of sensors is said to be in {\em attached positions}
if the union of their covering intervals is a
continuous interval on the $x$-axis whose length is equal to the
sum of the lengths of the covering intervals of these sensors.
A single sensor is always in attached position.
The following lemma
was proved in \cite{ref:CzyzowiczOn09}.

%
%
\begin{lemma}\label{lem:20}
{\em (Czyzowicz {\em et al.} \cite{ref:CzyzowiczOn09})}
If $\lambda^* > 0$, then in $OPT$, there exist a sequence of consecutive
sensors $s_i,s_{i+1},\ldots,s_j$ with $i\leq j$ such that they are in attached
positions and
one of the following three cases is true. (a) The sensor $s_j$
is moved to the left by the distance $\lambda^*$ and $y_i=r$
(i.e., the sensors $s_i,s_{i+1},\ldots,s_j$ together cover exactly the
interval $[0,2r(j-i+1)]$). (b) The sensor
$s_i$ is moved to the right by the distance $\lambda^*$ and $y_j=L-r$.
%
%
(c) For $i\neq j$ (i.e., $i<j$), the sensor
$s_i$ is moved to the right by the distance $\lambda^*$ and the sensor
$s_j$ is moved to the left by the distance $\lambda^*$.
\end{lemma}

Cases (a) and (b) in Lemma \ref{lem:20} are symmetric.
By Lemma \ref{lem:20}, for each pair of sensors $s_i$ and
$s_j$ with $i\leq j$, we can compute three distances
$\lambda_1(i,j), \lambda_2(i,j)$, and $\lambda_3(i,j)$ corresponding
to the three cases in Lemma \ref{lem:20} as
candidates for the optimal distance $\lambda^*$.
Specifically, $\lambda_1(i,j)=x_j-[2r(j-i)+r]$, where the value $2r(j-i)+r$ is
supposed to be the destination of the sensor $s_j$ in $OPT$ if case (a) holds.
Symmetrically, $\lambda_2(i,j)=[L-2r(j-i)-r]-x_i$. Let
$\lambda_3(i,j)=[x_j-x_i-2r(j-i)]/2$ for $i<j$.
Let $\Lambda$ be the set of all $\lambda_1(i,j), \lambda_2(i,j)$, and
$\lambda_3(i,j)$ values. Clearly, $\lambda^*\in \Lambda$ and
$|\Lambda|=\Theta(n^2)$.  By using an algorithm
for the decision version of the uniform case to search in $\Lambda$,
one can find the value $\lambda^*$. Recall that the decision problem is
that given any value $\lambda$, determine whether there exists a
feasible solution for covering $B$ such that the moving distances of all
sensors are at most $\lambda$. Thus, $\lambda^*$ is the smallest value
in $\Lambda$ such that the answer to the decision problem on that
value is ``yes". A simple greedy $O(n)$ time algorithm was given in
\cite{ref:CzyzowiczOn09} for the decision problem.

\begin{lemma}\label{lem:30}
{\em (Czyzowicz {\em et al.} \cite{ref:CzyzowiczOn09})}
The decision version of the uniform case is solvable in $O(n)$ time.
\end{lemma}

But, the above approach would take $\Omega(n^2)$ time due to
$|\Lambda|=\Theta(n^2)$. To reduce the running time, we cannot compute the set $\Lambda$ explicitly. In general, our $O(n\log n)$ time algorithm
uses the following idea. First, instead of computing all elements of
$\Lambda$ explicitly, we compute one element of
$\Lambda$ whenever we need it (we may do some preprocessing for
this). Second, suppose we already know (implicitly) a sorted order
of all values in $\Lambda$; then we can use binary search and the
decision algorithm for Lemma \ref{lem:30} to find $\lambda^*$. However,
we are not able to order the values of $\Lambda$ into a single sorted
list; instead, we order them (implicitly) in
$O(n)$ sorted lists and each list has $O(n)$ values. Consequently, by a
technique called {\em binary search in sorted arrays}
\cite{ref:ChenRe11}, we compute $\lambda^*$ in $O(n\log n)$ time.
The details of our algorithm are given in the next subsection.

\subsection{Our Algorithm for the Uniform Case}

Due to the order preserving property, it is easy to check
whether $\lambda^*=0$ in $O(n)$ time. In the
following, we assume $\lambda^*> 0$.

We focus on how to order (implicitly) the elements of $\Lambda$ into
$O(n)$ sorted lists and each list contains $O(n)$ elements. We also show that
after preprocessing, each element
in any sorted list can be computed in $O(1)$ time using
the index of the element. We aim to prove the next lemma.

\begin{lemma}\label{lem:sortedlist}
In $O(n\log n)$ time, the elements of $\Lambda$ can be ordered (implicitly)
into $O(n)$ sorted lists such that each list contains $O(n)$ elements and each
element
in any list can be computed in constant time by giving the index of the
list and the index of the element in the list.
\end{lemma}

The following technique, called {\em binary search on sorted
arrays} \cite{ref:ChenRe11}, will be applied.
Suppose there is a ``black-box"
decision procedure $\Pi$ available such that given any value $a$,
$\Pi$ can report whether $a$ is a {\em feasible value} to a certain problem
in $O(T)$ time, and further, if $a$ is a feasible value, then any
value larger than $a$ is also feasible. Given a set of $m$
arrays $A_i$, $1\leq i\leq m$, each containing $O(n)$ elements in
sorted order, the goal is to find the
smallest feasible value $\delta$ in $A=\cup_{i=1}^m A_i$.
Suppose each element of any array can be obtained in constant time
by giving its indices. An algorithm for the following result was presented
by Chen {\it et al.}~in \cite{ref:ChenRe11}.

\begin{lemma}\label{lem:binarysearch}
{\em (Chen {\em et al.} \cite{ref:ChenRe11})}
The value $\delta$
in $A$ can be computed in $O((m+T)\log (nm))$ time.
\end{lemma}

If we use the algorithm for Lemma \ref{lem:30} as the decision procedure $\Pi$, then
by Lemmas \ref{lem:sortedlist} and \ref{lem:binarysearch}, we can find
$\lambda^*$ in $\Lambda$ in $O(n\log n)$ time.
After $\lambda^*$ is found, we can apply the algorithm for Lemma
\ref{lem:30} to compute the destinations of all sensors, in $O(n)$
time.  Hence, we have the following result.

\begin{theorem}
The uniform case of the \bcls\ problem is solvable in $O(n\log n)$ time.
\end{theorem}

In the rest of this subsection, we focus on proving Lemma \ref{lem:sortedlist}.

For each $1\leq t\leq 3$, let $\Lambda_t$ denote the set of all
$\lambda_t(i,j)$ values. Clearly,
$\Lambda=\Lambda_1\cup \Lambda_2\cup \Lambda_3$. We seek to order each of
the three sets $\Lambda_1,\Lambda_2$, and $\Lambda_3$ into sorted lists.

We discuss $\Lambda_1$ first. This case is trivial. It is easy to see that
for each given value $j$, we
have $\lambda_1(i_1,j)\leq \lambda_1(i_2,j)$ for any $i_1\leq i_2\leq
j$. Thus, for every value $j$, we have a sorted list
$\lambda_1(1,j),\lambda_1(2,j),\ldots,\lambda_1(j,j)$ of $j$ elements,
and each element can be computed in constant time by using the index of
the element in the list. Therefore, we have $n$ sorted lists, and
clearly, the set of elements in all these lists is exactly
$\Lambda_1$. Hence we have the following lemma.

%
%
%

\begin{lemma}\label{lem:sortedlist1}
In $O(n\log n)$ time,
the elements of $\Lambda_1$ can be ordered (implicitly) into $O(n)$
sorted lists such that each list contains $O(n)$ elements and each element
in any list can be computed in constant time by giving the index of the
list and the index of the element in the list.
\end{lemma}

The set $\Lambda_2$ can be processed in a symmetric manner as
$\Lambda_1$, and we omit the details.

\begin{lemma}\label{lem:sortedlist2}
In $O(n\log n)$ time,
the elements of $\Lambda_2$ can be ordered (implicitly) into $O(n)$
sorted lists such that each list contains $O(n)$ elements and each element
in any list can be computed in constant time by giving the index of the
list and the index of the element in the list.
\end{lemma}

In the following, we focus on ordering (implicitly) the set $\Lambda_3$ and
showing the following lemma, which, together with Lemmas \ref{lem:sortedlist1}
and \ref{lem:sortedlist2}, proves Lemma \ref{lem:sortedlist}.

\begin{lemma}\label{lem:sortedlist3}
In $O(n\log n)$ time,
the elements of $\Lambda_3$ can be ordered (implicitly) into $O(n)$
sorted lists such that each list contains $O(n)$ elements and each element
in any list can be computed in constant time by giving the index of the
list and the index of the element in the list.
\end{lemma}

Proving Lemma \ref{lem:sortedlist3} is a main challenge to our uniform
case algorithm. The reason is that, unlike $\Lambda_1$ and $\Lambda_2$, for
a given $j$, for any $1\leq i_1\leq i_2\leq j$, either
$\lambda_3(i_1,j)\leq \lambda_3(i_2,j)$ or
$\lambda_3(i_1,j)\geq \lambda_3(i_2,j)$ is possible.  Hence, to prove
Lemma \ref{lem:sortedlist3}, we have to find another way to order the
elements of $\Lambda_3$.


Our approach is to first remove some elements from $\Lambda_3$ that
are surely not $\lambda^*$ (for example, negative values
cannot be $\lambda^*$). We begin with some intuitions. We say two
intervals on the $x$-axis are {\em strictly overlapped} if they contain more than one common point. In the following discussion, the sensors are
always at their input positions unless otherwise stated. We define
two subsets of sensors, $S_a$ and $S_b$, as follows. A sensor $s_j$ is
in $S_a$ if and only if there is no sensor $s_i$ with $i<j$ such that
their covering intervals are strictly overlapped (e.g., see
Fig.~\ref{fig:groups}). A sensor $s_i$ is in $S_b$ if
and only if there is no sensor $s_j$ with $i<j$ such that their
covering intervals are strictly overlapped. Let the indices of sensors in
$S_a$ be $a_1,a_2,\ldots,a_{n_1}$ and the indices of sensors in
$S_b$ be $b_1,b_2,\ldots,b_{n_2}$, from left to right.
We claim $n_1=n_2$. To see this,
consider the interval graph $G$ in which the covering interval of each
sensor is a vertex and two vertices are connected by an edge if their
corresponding intervals are strictly overlapped. Observe that in each
connected component of $G$, there is exactly one interval whose
corresponding sensor is in $S_a$ and there is exactly one interval
whose corresponding sensor is in $S_b$, and vice versa. Thus,
$n_1=n_2$, which is the number of connected components of $G$. Let
$m=n_1=n_2\leq n$. Further, it is easy to see that the covering intervals of
both $a_i$ and $b_i$ must be in the same connected component of $G$ and
$a_i\leq b_i$. Indeed, $a_i$ (resp., $b_i$) is the leftmost (resp.,
rightmost) sensor in the subset of sensors whose covering intervals
are in the same connected component of $G$ (see
Fig.~\ref{fig:groups}). Note that $a_i=b_i$ is
possible.  Hence, $G$ has $m$ connected components.

For each $1\leq i\leq m$,
let $G_i$ denote the connected component containing the
covering intervals of $a_i$ and $b_i$; with a little
abuse of notation, we also use $G_i$ to denote the subset of sensors
whose covering intervals are in the connected component $G_i$.
Clearly, $G_i=\{s_j\ |\ a_i\leq j\leq b_i\}$ (e.g., see
Fig.~\ref{fig:groups}). We also call $G_i$ a {\em group} of sensors.
The groups $G_1,G_2,\ldots,G_m$ form a partition of $S$.
The sensor $s_{a_i}$ (resp., $s_{b_i}$) is the leftmost
(resp., rightmost) sensor in $G_i$.


\begin{figure}[t]
\begin{minipage}[t]{\linewidth}
\begin{center}
\includegraphics[totalheight=0.8in]{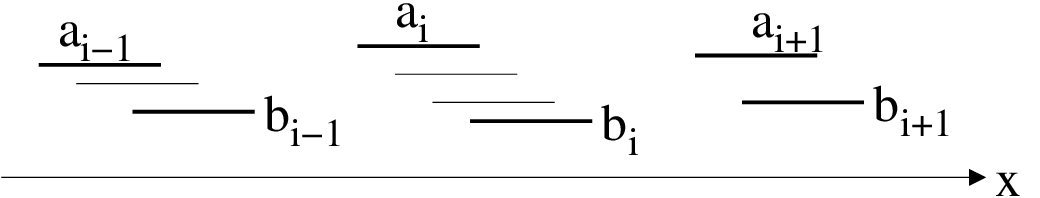}
\caption{\footnotesize Illustrating three groups of sensors:
$G_{i-1},G_i$, and $G_{i+1}$. The sensors with indices $a_{i-1},a_i$,
and $a_{i+1}$ (resp., $b_{i-1},b_i$, and $b_{i+1}$) are in $S_a$
(resp., $S_b$), which are indicated by thick segments.}
\label{fig:groups}
\end{center}
\end{minipage}
\vspace*{-0.15in}
\end{figure}

\begin{lemma}\label{lem:remove}
For any two sensors $s_i$ and $s_j$ with $i<j$, if $s_i\not\in S_b$ or
$s_j\not\in S_a$, then $\lambda_3(i,j)\neq\lambda^*$.
\end{lemma}
\begin{proof}
Assume $s_i\not\in S_b$. In the following, we prove that $\lambda_3(i,j)$
cannot be $\lambda^*$ for any $i<j$.

Suppose $s_i$ is in the group $G_k$. Then $i<b_k$ due to $s_i\not\in S_b$.
Further, the
covering intervals of $s_i$ and $s_{i+1}$ must be strictly overlapped
(otherwise, $s_i$ would be in $S_b$). Assume to the contrary
$\lambda_3(i,j)=\lambda^*$, which implies that case (c) in Lemma
\ref{lem:20} holds. Thus, in the corresponding $OPT$,
$s_i$ is moved to the right by the
distance $\lambda_3(i,j)$ and all sensors $s_i,s_{i+1},\ldots,s_j$
must be in attached positions. It is easy to see that the
sensor $s_{i+1}$ must move to the right by the distance
$\lambda_3(i,j)+2r-(x_{i+1}-x_i)$. Since the covering intervals of
$s_i$ and $s_{i+1}$ are strictly overlapped, $2r-(x_{i+1}-x_i)>0$.
Therefore, the moving distance of $s_{i+1}$ must be larger than that
of $s_i$. Since the moving distance of $s_i$ is
$\lambda_3(i,j)=\lambda^*$, we have contradiction. Hence, $\lambda_3(i,j)$
cannot be $\lambda^*$.

Assume $s_j\not\in S_a$. Then by a symmetric argument, we can
prove $\lambda_3(i,j)\not=\lambda^*$ for any $i<j$.
\end{proof}

By Lemma \ref{lem:remove}, if $\lambda^*\in\Lambda_3$, then it can only be in
the set $\Lambda_3'=\{\lambda_3(i,j)\ |\ i< j, i\in S_b, j\in S_a\}$,
and $|\Lambda_3'|=O(m^2)$. Thus, we seek to order the
elements of $\Lambda_3'$ into $O(m)$ sorted lists and each list contains
$O(m)$ elements. One might be tempting to use the following way. Clearly,
for each $1\leq k\leq m-1$, $\Lambda_3'$ contains
$\lambda_3(s_{b_k},s_{a_h})$ for all $h=k+1,k+2,\ldots, m$, and hence one may
simply put them into a list. However, such a list is not necessarily sorted.
Specifically, for any two indices $h_1$ and $h_2$ with $k+1\leq
h_1< h_2\leq m$, either $\lambda_3(s_{b_k},s_{a_{h_1}})\leq
\lambda_3(s_{b_k},s_{a_{h_2}})$ or $\lambda_3(s_{b_k},s_{a_{h_1}}) >
\lambda_3(s_{b_k},s_{a_{h_2}})$ is possible. Our approach relies on
additional observations. Below, for simplicity of notation, we use
$\lambda_3(b_k,a_h)$ to refer to $\lambda_3(s_{b_k},s_{a_h})$. We
first examine the value of each $\lambda_3(b_k,a_h)$ in $\Lambda_3'$.

By definition, we have
$\lambda_3(b_k,a_{k+1})=(x_{a_{k+1}}-x_{b_k}-2r)/2$, and this is equal to
half the length of the interval between the right extension of $s_{b_k}$ and
the left extension of $s_{a_{k+1}}$, which we call a {\em gap}.
Note that this gap is $\emptyset$ when the two sensors $s_{b_k}$ and
$s_{a_{k+1}}$ are in attached positions.
For each $1\leq k\leq m-1$, define $g_k=x_{a_{k+1}}-x_{b_k}-2r$,
which is the length of the corresponding gap. Hence,
$\lambda_3(b_k,a_{k+1})=g_k/2$.
Further, for each $1\leq k\leq m$, we define the {\em width} of the
group $G_k$ as the length of the union interval of the covering intervals of
the sensors in $G_k$,
and define $l_k$ as the sum of the lengths of the covering intervals of
the sensors in $G_k$ minus the width of $G_k$,
i.e., $l_k$ is equal to $2r(b_k-a_k+1)$ minus the width of $G_k$. We
then have the following observation.

\begin{observation}
For every $k$ with $1\leq k\leq  m-2$, we have
$\lambda_3(b_k,a_{h})=(\sum_{t=k}^{h-1}g_t-\sum_{t=k+1}^{h-1}l_t)/2$ for
each $h$ with $k+2\leq h\leq m$.
\end{observation}
\begin{proof}
By definition,
$\lambda_3(b_k,a_{h})=[(x_{a_h}-x_{b_k}-2r)-2r(a_h-b_k-1)]/2$. It is
easy to see that the value $x_{a_h}-x_{b_k}-2r$ is equal to
$\sum_{t=k}^{h-1}g_t$ plus the sum of the widths of all groups
$G_{k+1},G_{k+2},\ldots,G_{h-1}$, and the value $2r(a_h-b_k-1)$ is
equal to the sum of the lengths of the covering intervals of the
sensors in the union of the groups $G_{k+1},G_{k+2},\ldots,G_{h-1}$.
According to definitions of $l_t$ for $k+1\leq t\leq h-1$, the
observation follows.
\end{proof}

The following lemma will be useful later.

%
%
%
%

\begin{lemma}\label{lem:relativeorder}
For four indices $k_1,k_2,h_1$, and $h_2$, suppose
$\max\{k_1,k_2\}<\min\{h_1,h_2\}$; then
$\lambda_3(b_{k_1},a_{h_1})$ $- \lambda_3(b_{k_1},a_{h_2})=
\lambda_3(b_{k_2},a_{h_1})- \lambda_3(b_{k_2},a_{h_2})$, and
consequently,
$\lambda_3(b_{k_1},a_{h_1})\leq \lambda_3(b_{k_1},a_{h_2})$ if and
only if $\lambda_3(b_{k_2},a_{h_1})\leq \lambda_3(b_{k_2},a_{h_2})$.
\end{lemma}
\begin{proof}
Note that for every $1\leq k\leq m-2$, we have
$\lambda_3(b_k,a_{h})=(\sum_{t=k}^{h-1}g_t-\sum_{t=k+1}^{h-1}l_t)/2$
for $k+2\leq h\leq m$, and $\lambda_3(b_k,a_{h})=g_k/2$ for $h=k+1$.

If $h_1=h_2$, then the lemma trivially follows since
$\lambda_3(b_{k_1},a_{h_1})=\lambda_3(b_{k_1},a_{h_2})$ and
$\lambda_3(b_{k_2},a_{h_1})=\lambda_3(b_{k_2},a_{h_2})$.
Thus we consider $h_1\neq h_2$, and only show the case with $h_1< h_2$
(the case with $h_1> h_2$ is similar).
%
%
%
By their definitions, we have
$\lambda_3(b_{k_1},a_{h_1})-\lambda_3(b_{k_1},a_{h_2})=
(-\sum_{t=h_1}^{h_2-1}g_t+\sum_{t=h_1}^{h_2-1}l_t)/2$. Similarly,
$\lambda_3(b_{k_2},a_{h_1})-\lambda_3(b_{k_2},a_{h_2})=
(-\sum_{t=h_1}^{h_2-1}g_t+\sum_{t=h_1}^{h_2-1}l_t)/2$.
Hence, the lemma follows.
\end{proof}

Lemma \ref{lem:relativeorder} implies that for any $k_1$ and $k_2$
with $1\leq k_1< k_2\leq m-1$, the sorted order of
$\lambda_3(b_{k_1},a_t)$ for all $t=k_2+1,k_2+2,\ldots,m$ is the same as
that of the list $\lambda_3(b_{k_2},a_t)$ for
$t=k_2+1,k_2+2,\ldots,m$ in terms of the indices of $a_t$.
This means that if we sort the values
in the list $\lambda_3(b_1,a_t)$ for all $t=2,3,\ldots,m$, then for
any $1<k\leq m-1$, the sorted order of the list $\lambda_3(b_k,a_t)$ with all
$t=k+1,k+2,\ldots,m$ is also obtained implicitly. Our ``ordering" algorithm
works as follows.

We first explicitly compute the values $\lambda_3(b_1,a_t)$ for
all $t=2,3,\ldots,m$, which takes $O(m)$ time, and then sort them in
$O(m\log m)$ time. Let $p$ be the permutation of $2,3,\ldots,m$ such
that the increasing sorted list of these $\lambda_3(b_1,a_t)$
values is $\lambda_3(b_1,a_{p(1)}),
\lambda_3(b_1,a_{p(2)}),\ldots,\lambda_3(b_1,a_{p(m-1)})$. Note that
the permutation $p$ is immediately available once we obtain the above
sorted list. For any $1<k<m$, we say the element
$\lambda_3(b_k,a_h)$ is {\em valid} if $k+1\leq h\leq m$ and is {\em
undefined} otherwise.
By Lemma \ref{lem:relativeorder}, the valid elements in each list
$\lambda_3(b_k,a_{p(1)}),
\lambda_3(b_k,a_{p(2)}),\ldots,\lambda_3(b_k,a_{p(m-1)})$ are also
sorted increasingly. Further, if we compute $g_1,g_2,\ldots,g_{m-1}$
and $l_1,l_2,\ldots,l_m$ as well as their prefix sums in the preprocessing,
then given the index of any valid element in the list, we can
obtain its actual value in $O(1)$ time. Clearly, the preprocessing
takes $O(n\log n)$ time.  Thus, we have ordered (implicitly) the
elements of $\Lambda'_3$ into $O(m)$ sorted lists and each list has $O(m)$
elements.

However, we are not done yet. Since eventually we will apply the binary
search technique of Lemma \ref{lem:binarysearch} to these sorted lists and
the lists contain undefined elements, the algorithm may take an undefined
element in such a list and use it in the decision procedure which is
the algorithm for Lemma \ref{lem:30}. But, the undefined elements do not
have meaningful values. To resolve
this, we assign (implicitly) to each undefined element an
``appropriate" value, as follows.
For each $1\leq k<m$, let $\calL(k)$ denote the list
$\lambda_3(b_k,a_{p(1)}), \lambda_3(b_k,a_{p(2)}),\ldots,\lambda_3(b_k,a_{p(m-1)})$.
If $1<k<m$, then the list $\calL(k)$ has some undefined elements. For each
undefined element, we (implicitly) assign an actual
value to it such that the resulting new list is still sorted.
The idea is inspired by Lemma \ref{lem:relativeorder}.  We use the
list $\calL(1)$ as the {\em reference list} since all its elements are
valid. Every other list $\calL(k)$ has at least one valid element, for
example, the element $\lambda_3(b_k,a_{k+1})$.
We compute explicitly the value $\lambda_3(b_k,a_{k+1})$ for each $1<k<m$,
in $O(m)$ time. For a list $\calL(k)$ with $1<k<m$ and any
undefined element $\lambda_3(b_k,a_{p(i)})$ in $\calL(k)$, we assign
to it (implicitly) the value $\lambda_3(b_k,a_{k+1})+\lambda_3(b_1,a_{p(i)})-
\lambda_3(b_1,a_{k+1})$ (note that all these three values have already
been computed explicitly).
The lemma below shows that the resulting new list $\calL(k)$ is still sorted
increasingly with this value assignment scheme.

\begin{lemma}\label{lem:newlist}
For any $1< k<m$, the list $\calL(k)$ is still sorted increasingly after
all its undefined elements are assigned values implicitly.
\end{lemma}
\begin{proof}
Consider any $k$ with $1< k<m$, and any two indices $i$ and $j$ with
$1\leq i< j\leq m-1$. It is sufficient to prove
$\lambda_3(b_k,a_{p(i)})\leq \lambda_3(b_k,a_{p(j)})$.

If both values are valid, then by Lemma \ref{lem:relativeorder},
the inequality holds.
Otherwise, we assume $\lambda_3(b_k,a_{p(i)})$ is undefined.
After our value assignment, $\lambda_3(b_k,a_{p(i)})=
\lambda_3(b_k,a_{k+1})+\lambda_3(b_1,a_{p(i)})-\lambda_3(b_1,a_{k+1})$.
Depending on whether $\lambda_3(b_k,a_{p(j)})$ is undefined,
there are two cases.

If $\lambda_3(b_k,a_{p(j)})$ is undefined, then we have
$\lambda_3(b_k,a_{p(j)})=
\lambda_3(b_k,a_{k+1})+\lambda_3(b_1,a_{p(j)})-\lambda_3(b_1,a_{k+1})$.
Hence, $\lambda_3(b_k,a_{p(j)})-\lambda_3(b_k,a_{p(i)})=
\lambda_3(b_1,a_{p(j)})-\lambda_3(b_1,a_{p(i)})\geq 0$ due to $j>i$.
If $\lambda_3(b_k,a_{p(j)})$ is valid, then by Lemma \ref{lem:relativeorder},
we have $\lambda_3(b_k,a_{p(j)})=
\lambda_3(b_k,a_{k+1})+\lambda_3(b_1,a_{p(j)})-\lambda_3(b_1,a_{k+1})$.
Thus, $\lambda_3(b_k,a_{p(j)})-\lambda_3(b_k,a_{p(i)})=
\lambda_3(b_1,a_{p(j)})-\lambda_3(b_1,a_{p(i)})\geq 0$.

Therefore, in both cases, we have $\lambda_3(b_k,a_{p(i)})\leq
\lambda_3(b_k,a_{p(j)})$, which proves the lemma.
\end{proof}

In summary, in $O(n\log n)$ time, we have (implicitly) ordered the
elements of $\Lambda'_3$ into $O(m)$ sorted lists and each list has $O(m)$
elements such that every element in any list can be obtained in
$O(1)$ time. Hence, Lemma \ref{lem:sortedlist3} is proved.
We remark that assigning values to the undefined elements in $\Lambda'_3$
as above does not affect the correctness of our algorithm. Assigning
values to undefined elements only makes our candidate set $\Lambda$
for $\lambda^*$ a little larger (by a constant factor), which obviously
does not affect the algorithm
correctness because the larger candidate set still contains $\lambda^*$.
One might also see that the statement of Lemma \ref{lem:sortedlist3}
(and thus Lemma \ref{lem:sortedlist}) is a little imprecise since we
actually ordered only the elements in a subset $\Lambda_3'$ of $\Lambda_3$
(not the entire set $\Lambda_3$).

\subsection{The Special Uniform Case}

In this subsection, we consider the special uniform case in which
all sensors are initially located on the barrier $B=[0,L]$, i.e.,
$0\leq x_i\leq L$ for each $1\leq i\leq n$. We give an
$O(n)$ time algorithm for it. Again, we assume $\lambda^*> 0$.

Clearly, Lemmas \ref{lem:10} and \ref{lem:20} still hold. Further,
since all sensors are initially on $B$,
in case (a) of Lemma \ref{lem:20}, $s_i$ must be $s_1$. To
see this, since $s_1$ is initially located on $B=[0,L]$,
it is always the best to use $s_1$ to cover the
beginning portion of $B$ due to the order preserving property. We omit
the formal proof of this.
Similarly, in case (b) of Lemma \ref{lem:20}, $s_j$ must be $s_n$. We
restate Lemma \ref{lem:20} below as a corollary for this special case.

\begin{corollary}\label{cor:10}
If $\lambda^* > 0$, then in $OPT$, there exist a sequence of consecutive
sensors $s_i,s_{i+1},\ldots,s_j$ with $i\leq j$ such that they
are in attached positions and
one of the following three cases is true. (a) The sensor
$s_j$ is moved to the left by the distance $\lambda^*$, $i=1$, and
$y_1=r$. (b) The sensor
$s_i$ is moved to the right by the distance $\lambda^*$, $j=n$, and
$y_n=L-r$. (c) For $i\not=j$ (i.e., $i<j$), the sensor
$s_i$ is moved to the right by the distance $\lambda^*$ and the sensor
$s_j$ is moved to the left by the distance $\lambda^*$.
\end{corollary}

For any $1\leq i<j\leq n$, we define $\lambda_3(i,j)$ in the
same way as before, i.e., $\lambda_3(i,j)=[x_j-x_i-2r(j-i)]/2$,
which corresponds to case (c) of Corollary \ref{cor:10}.
For each $1\leq j\leq n$, define
$\lambda_1'(j)=x_j+r-2rj$, which corresponds to case (a).
Similarly, for each $1\leq i\leq n$, define
$\lambda_2'(i)=L-2r(n-i)-(x_i+r)$, which corresponds to case (b).
We still use $\Lambda_3$ to denote the set of all $\lambda_3(i,j)$
values. Define $\Lambda'_1=\{\lambda'_1(j)\ |\ 1\leq j\leq n\}$ and
$\Lambda'_2=\{\lambda'_2(i)\ |\ 1\leq i\leq n\}$. Let
$\Lambda'=\Lambda'_1\cup\Lambda'_2\cup\Lambda_3$.
By Corollary \ref{cor:10}, we have $\lambda^*\in \Lambda'$.
The following lemma is crucial to our algorithm.

\begin{lemma}\label{lem:specialcase}
The optimal value $\lambda^*$ is the maximum value in $\Lambda'$.
\end{lemma}
\begin{proof}
Let $\lambda'$ be the maximum value in
$\Lambda'$. It suffices to show
$\lambda^*\leq \lambda'$ and $\lambda'\leq \lambda^*$.
Since $\lambda^*\in \Lambda'$, $\lambda^*\leq
\lambda'$ trivially holds. Below, we focus on proving
$\lambda'\leq \lambda^*$.

Since $\lambda^*>0$, $\lambda'>0$ holds.
Clearly, either $\lambda'\in\Lambda'_1$, or $\lambda'\in\Lambda'_2$, or
$\lambda'\in\Lambda_3$. Below we analyze these three cases.

If $\lambda'\in\Lambda'_1$, then suppose $\lambda'=\lambda_1'(j)$ for some
$j$. Since $\lambda'>0$, we have
$0<\lambda'=\lambda_1'(j)=x_j+r-2rj$, and thus $x_j-r>2r(j-1)$. Since
all sensors are initially on the barrier $B$, $x_j\leq L$ holds.
Hence, even if all sensors $s_1,s_2,\ldots,s_{j-1}$ are somehow
moved such that they are in attached positions to cover the sub-interval
$[0,2r(j-1)]$ of $B$, the sub-interval $[2r(j-1),x_j-r]$ of $B$ is still not
covered by any of the sensors $s_1,s_2,\ldots,s_{j-1}$. By the
order preserving property, to cover the sub-interval $[2r(j-1),x_j-r]$,
the best way is to move $s_j$ to the left such that the new position of $s_j$
is at $2r(j-1)+r$ (i.e., the sensors $s_1,s_2,\ldots,s_{j}$ are in attached
positions), for which the moving distance of $s_j$ is exactly
$\lambda_1'(j)$.
Therefore, the maximum sensor movement in any optimal solution has to
be at least $\lambda_1'(j)$. Thus, $\lambda'=\lambda_1'(j)\leq \lambda^*$.

If $\lambda'\in\Lambda'_2$, then the analysis is symmetric to the above
case and we omit the details.

When $\lambda'\in\Lambda_3$, the analysis has a similar spirit
and we briefly discuss it. Suppose
$\lambda'=\lambda_3(i,j)=[x_j-x_i-2r(j-i)]/2$ for some $i<j$. Since all
sensors are initially on the barrier $B$, we have $0\leq x_i<x_j\leq
L$. Consider the sub-interval $[x_i+r,x_j-r]$ of $B$. Because
$\lambda'>0$, we have $x_j-x_i-2r(j-i)>0$, and thus
$(x_j-r)-(x_i+r)>2r(j-i-1)$. This implies that even if we somehow move the
sensors $s_{i+1},s_{i+2},\ldots,s_{j-1}$ such that they are in
attached positions inside $[x_i+r,x_j-r]$, there are still points in
$[x_i+r,x_j-r]$ that are not covered by the sensors
$s_{i+1},s_{i+2},\ldots,s_{j-1}$. By
the order preserving property, to cover the interval
$[x_i+r,x_j-r]$, we have to use both $s_i$ and $s_j$ and the best way
is to move $s_i$ to the right and move $s_j$ to the left by an equal
distance so that all sensors $s_{i},s_{i+1},\ldots,s_{j}$ are in
attached positions, for which the moving distances of $s_i$ and $s_j$ are
both $\lambda_3(i,j)$ exactly.  Therefore, the maximum
sensor movement in any optimal solution has to be at least
$\lambda_3(i,j)$. Thus, $\lambda'=\lambda_3(i,j)\leq \lambda^*$.

In summary, in any case, $\lambda'\leq \lambda^*$ holds.
The lemma thus follows.
\end{proof}


Base on Lemma \ref{lem:specialcase}, to compute $\lambda^*$,
we only need to find the maximum value in $\Lambda'$, which can be easily
obtained in $O(n^2)$ time by computing the set $\Lambda'$ explicitly
(note that $|\Lambda'|=\Theta(n^2)$). Yet, we show below that we can find
its maximum value in $O(n)$ time without computing $\Lambda'$ explicitly.

\begin{lemma}\label{lem:algospecialcase}
The maximum value in $\Lambda'$ can be computed in $O(n)$ time.
\end{lemma}
\begin{proof}
Let $\lambda_1$, $\lambda_2$, and $\lambda_3$ be the maximum values in
the three sets $\Lambda_1'$, $\Lambda_2'$, and $\Lambda_3$,
respectively. It is sufficient to show how to compute $\lambda_1$, $\lambda_2$,
and $\lambda_3$ in $O(n)$ time.

Both the sets $\Lambda_1'$ and $\Lambda_2'$ can be computed explicitly in
$O(n)$ time. Thus, we can find $\lambda_1$ and $\lambda_2$ in
$O(n)$ time. Below, we focus on computing $\lambda_3$.

Note that for each value $\lambda_3(i,j)\in \Lambda_3$ with $i<j$,
we have $\lambda_3(i,j)=[x_j-x_i-2r(j-i)]/2$. For each $1\leq t\leq
n-1$, define $z_t=x_{t+1}-x_t-2r$. Hence,
$\lambda_3(i,j)=(\sum_{t=i}^{j-1}z_t)/2$.  This implies that finding
the maximum value in $\Lambda_3$ is equivalent to
finding a consecutive subsequence of $z_1,z_2,\ldots,z_{n-1}$
such that the sum of the subsequence is the maximum among all possible
consecutive subsequences, which is an instance of
the well studied {\it maximum subsequence sum problem}.  This problem can
be solved easily in $O(n)$ time. Specifically, we first compute all
values $z_1,z_2,\ldots,z_{n-1}$, in $O(n)$ time. If all values are
negative, then $\lambda_3$ is the maximum value divided by $2$.
Otherwise, we let $z'_0=0$, and for
each $1\leq t\leq n-1$, let $z'_t=\max\{z'_{t-1},0\}+z_t$. It is not
difficult to see that $\lambda_3=\frac{1}{2}\cdot \max_{1\leq t\leq n-1}\{z'_t\}$.
Hence, $\lambda_3$ can be computed in $O(n)$ time.
The lemma thus follows.
\end{proof}

After $\lambda^*$ is computed, we use the linear time decision algorithm
for Lemma \ref{lem:30} to compute the destinations of all sensors such
that the maximum sensor movement is at most $\lambda^*$.

\begin{theorem}
The special uniform case of the \bcls\ problem is solvable in $O(n)$ time.
\end{theorem}

\section{The Simple Cycle Barrier Coverage}
\label{sec:circle}

In this section, we discuss the simple cycle barrier coverage problem
and present an $O(n)$ time algorithm for it.
Mehrandish \cite{ref:MehrandishOn11} gave an $O(n^2)$ time algorithm by somehow
generalizing the $O(n^2)$ time algorithm \cite{ref:CzyzowiczOn09} for the uniform
\bcls.

In this problem, the target region $R$ is on the plane
enclosed by a simple cycle $B$ that is the barrier we aim to cover.
The sensors in $S=\{s_1,s_2,\ldots,s_n\}$ are initially located on $B$ and
each sensor is allowed to
move only on $B$ (e.g., not allowed to move inside or outside $R$). All
sensors in $S$ have the same range $r$. Here, the {\em distance} between any
two points on $B$ is not measured by their Euclidean distance in the plane but by
their shortest distance along $B$. If a sensor is at a point
$p$ on $B$, then it covers all points of $B$ whose distances to $p$ are
at most $r$. Suppose all sensors in $S$ are initially ordered clockwise on
$B$ as specified by their indices.
Our goal is to move the sensors along $B$ to form a coverage of $B$
such that the maximum sensor movement is minimized.

Since $B$ is a cycle here, a sensor is said to
move clockwise or counterclockwise (instead of right or left).
Let $L$ be the length of $B$. Again, we assume $L\leq 2nr$
(otherwise, it would not be possible to form a coverage of $B$).
Since $B$ is a cycle, if $L\leq 2r$, then every sensor by itself forms a
coverage of $B$. Below, we assume $L>2r$. Imagine that we pick a
point $p_0$ on the interval of $B$ from $s_n$
clockwise to $s_1$ as the {\it origin} of $B$, and define
the {\em coordinate} of each point $p \in B$ as the distance traversed as
we move from $p_0$ to $p$ clockwise along $B$. Let the input coordinate
of each sensor $s_i\in S$ be $x_i$. Thus, we have $0<x_1\leq x_2\leq
\cdots\leq x_n<L$. Further, for each $1\leq
i\leq n$, we let $s_{i+n}$ denote a {\it duplication} of the sensor $s_i$
with a {\em coordinate} $x_{i+n}=x_i+L$, which actually refers to the
position on $B$ with the coordinate $x_i$.

Since all sensors have the same range, it is easy to see that there
always exists an {\em order preserving} optimal solution $OPT$ in which
the sensors are ordered clockwise along $B$ in the same order as that of their
input indices. A formal proof for this is given in \cite{ref:MehrandishOn11}.
Again, let $\lambda^*$ be the optimal moving distance.  We
can check whether $\lambda^*=0$ in $O(n)$ time. Below, we assume $\lambda^*>0$.

Actually, our algorithm considers a set of $2n$ sensors,
$S'=\{s_1,s_2,\ldots,s_{2n}\}$. Specifically, the algorithm determines a
consecutive sequence of sensors, $S_{ij}'=\{s_i,s_{i+1},\ldots,s_j\}\subset S'$
with $1\leq i< j< i+n$, and moves the sensors of $S_{ij}'$ to form a barrier
coverage of $B$ such that the maximum sensor movement is minimized. Clearly,
for each sensor $s_k\in S$, at most one of $s_k$ and its duplication
$s_{k+n}$ is in $S_{ij}'$.
In this simple cycle case,
the definition of {\em attached positions} of the sensors is
slightly different from that in the line segment case.
In this case, we call the two endpoints of the covering interval of
a sensor $s_i$ its {\em counterclockwise} and {\em clockwise extensions},
such that when going clockwise from the counterclockwise extension to the
clockwise extension, we move along the covering interval of $s_i$.
One sensor is always in attached position by itself.  Two
sensors are in {\em attached positions} if the clockwise
extension of one sensor is at the same position as the counterclockwise
extension of the other sensor. Note that unlike in the line
segment case, if two sensors
$s_i$ and $s_j$ are in attached positions, say, the clockwise
extension of $s_i$ is at the same position as the counterclockwise extension
of $s_j$, then since the sensors are on the cycle, it is
possible that the clockwise extension of $s_j$ is in the interior of the
covering interval of $s_i$ (e.g., when $L<4r$).
Similarly, a sequence of sensors $s_i,s_{i+1},\ldots,s_j$ (with $1\leq
i<j<i+n$)
are in {\em attached positions} if the
clockwise extension of $s_t$ is at the same position as the
counterclockwise extension of $s_{t+1}$ for each $i\leq t\leq j-1$ (and
the clockwise extension of $s_j$ may be in the interior of the covering
interval of $s_i$). The next result is a corollary of Lemma \ref{lem:20}.

\begin{corollary}\label{cor:circular}
If $\lambda^*> 0$, then in $OPT$, there exist
a sequence of sensors $s_i,s_{i+1},\ldots,s_j$ in $S'$
with $1\leq i<j<i+n$ such that they are in attached positions and the sensor
$s_i$ is moved clockwise by the distance $\lambda^*$ and the sensor
$s_j$ is moved counterclockwise by the distance $\lambda^*$.
\end{corollary}

For each pair of $i$ and $j$ with $1\leq i<j< i+n$, we define
$\lambda(i,j)=[x_j-x_i-2r(j-i)]/2$. Let $\Lambda$ be the set of all
such $\lambda(i,j)$ values. By Corollary \ref{cor:circular},
$\lambda^*\in\Lambda$.
The following result is similar to that of Lemma \ref{lem:specialcase}.

\begin{lemma}\label{lem:circularcase}
The optimal value $\lambda^*$ is the maximum value in $\Lambda$.
\end{lemma}
\begin{proof}
The proof is very similar to that for Lemma \ref{lem:specialcase} and
we briefly discuss it below.

Let $\lambda'$ be the maximum value in $\Lambda$. It is sufficient to show
$\lambda^*\leq \lambda'$ and $\lambda'\leq \lambda^*$.
Due to $\lambda^*\in \Lambda$, $\lambda^*\leq \lambda'$ trivially holds.
Hence, we focus on proving $\lambda'\leq \lambda^*$.
Since $\lambda^*> 0$, we have $\lambda'>0$.

Suppose $\lambda'=\lambda(i,j)=[x_j-x_i-2r(j-i)]/2$ for some $i$ and
$j$ with $1\leq i<j<i+n$.
Consider the clockwise interval $[x_i+r,x_j-r]$ on $B$, i.e., the
union of the points on $B$ from $x_i+r$ to $x_j-r$ clockwise. Since
$\lambda'>0$, we have $x_j-x_i-2r(j-i)>0$, and thus
$(x_j-r)-(x_i+r)>2r(j-i-1)$. This implies that even if we somehow move the
sensors $s_{i+1},s_{i+2},\ldots,s_{j-1}$ such that they are in
attached positions inside $[x_i+r,x_j-r]$, there are still points in
$[x_i+r,x_j-r]$ that are not covered by the sensors
$s_{i+1},s_{i+2},\ldots,s_{j-1}$. By the order
preserving property, to cover the interval
$[x_i+r,x_j-r]$ on $B$, the best way is to move
$s_i$ clockwise and move $s_j$ counterclockwise by the same
distance, for which the moving distances of $s_i$ and $s_j$ are both
$\lambda(i,j)$ exactly.  Therefore, the maximum
sensor movement in any optimal solution has to be at least
$\lambda(i,j)$. Thus, $\lambda'=\lambda(i,j)\leq \lambda^*$.

The lemma thus follows.
\end{proof}

By using the same algorithm for Lemma \ref{lem:algospecialcase}, we can find
$\lambda^*$ in $\Lambda$ in $O(n)$ time. With the value $\lambda^*$,
we can then easily compute an optimal solution (i.e., compute the
destinations of all sensors) in $O(n)$ time, as follows.

Suppose $\lambda^*=\lambda(i,j)\in\Lambda$ for some $i$ and $j$ with
$1\leq i<j<i+n$. In the
case of $i>n$, we have $j>n$ and let $i=i-n$ and $j=j-n$. Thus, we
still have $\lambda^*=\lambda(i,j)$ since
$\lambda(i,j)=\lambda(i-n,j-n)$ when $i>n$ and $j>n$. Below, we assume
$1\leq i\leq n$. Note that $j>n$ is possible.

First, we move $s_i$ clockwise by the distance $\lambda^*$ and
move $s_j$ counterclockwise by the same distance $\lambda^*$.
Next, move all sensors $s_{i+1},s_{i+2},\ldots,s_{j-1}$
such that the sensors $s_i,s_{i+1},\ldots,s_{j}$ are in attached
positions. Since $\lambda^*$ is the maximum value in $\Lambda$
by Lemma \ref{lem:circularcase}, the
above movements of the sensors $s_{i+1},s_{i+2},\ldots,s_{j-1}$ are at
most $\lambda^*$.
Then, starting at the sensor $s_{j+1}$, we consider the other
sensors $s_{j+1},s_{j+2},\ldots,s_{i-1}$ of $S$ clockwise along $B$,
and move them to cover the portion
of $B$ that is not covered by the sensors
$s_i,s_{i+1},\ldots,s_j$. For this, we can view the remaining uncovered
portion of $B$ as a line segment and apply the linear time greedy
algorithm for Lemma \ref{lem:30} with the value $\lambda^*$. The overall running
time is $O(n)$.

\begin{theorem}
The simple cycle barrier coverage problem is solvable in $O(n)$ time.
\end{theorem}

Note that the case $L<2nr$ is also discussed in \cite{ref:MehrandishOn11},
where an issue of {\em balance points} appears. In the case
$L\geq 2nr$ that we consider, the issue does not exist because the
entire cycle $B$ must be covered by the sensors in the optimal solution.

\section{Conclusions}
\label{sec:conclusions}

We present several algorithms on minimizing the maximum
sensor movement for barrier coverage on linear domains. We
present the first-known polygonal time algorithm for the problem where
the barrier is a line segment and the sensors have different sensing
ranges, and the algorithm runs in $O(n^2\log n)$ time. If the sensing
ranges are the same, we give an $O(n\log n)$ time solution, and
further if the sensors are initially located on the barrier segment, then the
algorithm runs in $O(n)$ time. In addition, if the barrier is a simple
cycle and the sensing ranges are the same, our approach can solve the
problem in $O(n)$ time. Each of these results either is first-known or
improves the previous best-known algorithms.

An interesting question is whether our $O(n^2\log n)$ time algorithm
can be improved. One possible direction, as we did for the uniform
case in Section \ref{sec:uniform}, is to try to determine a set
$\Lambda$ of candidate values such that $\lambda^*\in \Lambda$, and
then use our decision algorithms given in Section \ref{sec:decision}
to find $\lambda^*$ in $\Lambda$. It would also be interesting to see
whether our techniques can be extended to solve the non-uniform simple
cycle case where the sensing ranges are different. In addition,
since when all sensors are initially located on the barrier
segment the uniform case can be solved faster, it is natural to ask
whether our $O(n^2\log n)$ time algorithm for the non-uniform case can be
made faster if all sensors are initially located on the barrier
segment.

%

\footnotesize
\baselineskip=11.0pt
\bibliographystyle{plain}
\bibliography{reference}


\end{document}